\documentclass[journal,draftcls,onecolumn,12pt,twoside]{IEEEtran}
\usepackage[T1]{fontenc}
\usepackage{cite}
\pdfoutput=1

\usepackage[dvipsnames]{xcolor}

\ifCLASSINFOpdf
\else
\fi
\usepackage{amsthm}
\usepackage{amsmath}
\usepackage{amsfonts}
\usepackage[cmintegrals]{newtxmath}

\usepackage{array}

\usepackage{stfloats}

\usepackage{bbm}
\usepackage[pdftex]{graphicx}

\usepackage{caption}
\usepackage{siunitx}
\usepackage{tikz}
\usepackage{pgfplots}
\usepackage{verbatim}
\usepackage{color,soul}
\usetikzlibrary{spy,calc}
\newtheorem{theorem}{Theorem}
\newtheorem{lemma}{Lemma}

\usetikzlibrary{calc,positioning}
\usetikzlibrary{plotmarks}
\usepackage{multirow,bigdelim,dcolumn,booktabs}
\usepackage{hhline}
\usepackage{colortbl}
\begin{document}
	
	\title{The Exact Rate Memory Tradeoff for Large Caches with Coded Placement }
	
	\author{Vijith Kumar K P,~
		Brijesh Kumar Rai~and~Tony Jacob}

	\maketitle

	\begin{abstract}
		The idea of coded caching for content distribution networks was introduced by Maddah-Ali and Niesen, who considered the canonical $(N, K)$ cache network in which a server with $N$ files satisfy the demands of $K$ users (equipped with independent caches of size $M$ each). Among other results, their work provided a characterization of the exact rate memory tradeoff for the problem when $M\geq\frac{N}{K}(K-1)$. In this paper, we improve this result for large caches with $M\geq \frac{N}{K}(K-2)$. For the case $\big\lceil\frac{K+1}{2}\big\rceil\leq N \leq K$, we propose a new coded caching scheme, and derive a matching lower bound to show that the proposed scheme is optimal. This extends the characterization of the exact rate memory tradeoff to the case $M\geq \frac{N}{K}\Big(K-2+\frac{(K-2+1/N)}{(K-1)}\Big)$. For the case $1\leq N\leq \big\lceil\frac{K+1}{2}\big\rceil$, we derive a new lower bound, which demonstrates that the scheme proposed by Yu et al. is optimal and thus extend the characterization of the exact rate memory tradeoff to the case $M\geq \frac{N}{K}(K-2)$. 
		    
	\end{abstract}


	\IEEEpeerreviewmaketitle

	\section{Introduction}
	
	In content distribution networks, some parts of the files in the server are placed in caches distributed across the network during the off-peak traffic time, so as to reduce the load experienced by the network during the peak traffic time. Maddah-Ali and Niesen, in their seminal work \cite{maddah2014fundamental}, noted that traditional caching schemes fail to exploit the multicast coding opportunity available in such networks. They introduced the notion of coded caching in the setting of the canonical $(N,K)$ cache network shown in Fig \ref{fig:canonical}.
	\begin{figure}[ht]
\begin{center}
	\input{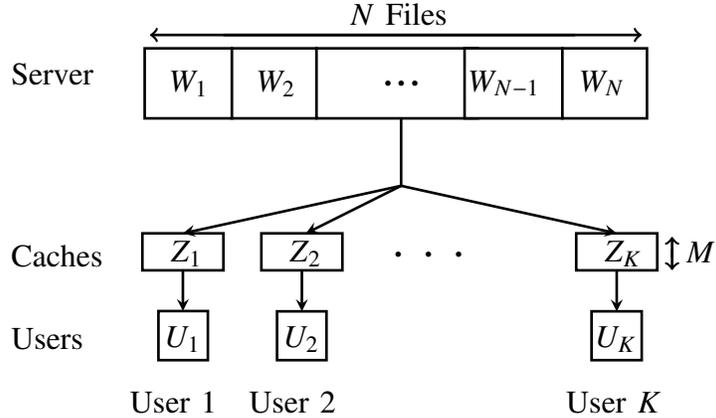}
	\caption{The canonical cache network considered in \cite{maddah2014fundamental}}
	\label{fig:canonical}
\end{center}
		
	\end{figure}
	The server has $N$ files $\{W_{1},\dots, W_{N}\}$ and is connected through an error free broadcast link to $K$ users $\{U_{1},\dots, U_{K}\}$, each with an isolated cache of size $M\in [0,N]$. The cache attached to user $U_{k}$ is denoted by $Z_{k}$. In the first phase of a coded caching scheme, called the placement phase, the server copies some fragments of the files available to it into the caches, without any knowledge of the files that will be required by each user. Let the demands by the users be represented by a vector $\textbf{\textit{d}}=(W_{d_{1}},\dots, W_{d_{K}})$, where $W_{d_{l}}$ is the file requested user $U_{l}$. In the second phase, called the delivery phase, the server broadcasts a set of packets $X_{\textbf{\textit{d}}}$ of size $R_{\textbf{\textit{d}}}(M)$ in response to the demand $\textbf{\textit{d}}$. Each user recovers its required file from the broadcast packets aided by the  contents of its isolated cache. The design of a coded caching scheme involves deciding what to place in the cache attached to each user  during the placement phase and what to broadcast for each possible demand such that the shared link experiences the minimum load during the delivery phase. This formulation of the problem of coded caching has been extended in several ways to study decentralised cache networks \cite{maddah2015decentralized}, hierarchical cache networks \cite{karamchandani2016hierarchical}, cache networks with multiple servers \cite{shariatpanahi2016multi}, coded caching with privacy \cite{ravindrakumar2017private}, heterogeneous cache networks \cite{daniel2019optimization}, networks with shared cache \cite{parrinello2019coded}, cache aided D2D networks \cite{yapar2019optimality} and data shuffling problems with cache aided worker nodes \cite{wan2018fundamental}.

	 In the case of the canonical $(N,K)$ cache network, all files and caches are of the same size. Due to the inherent symmetry of the problem, it is natural to group together all demands that are related to each other through a permutation. In \cite{tian2016symmetry}, Tian showed that corresponding to any caching scheme there exists a symmetric caching scheme which operates with the same or smaller rate. Hence we consider only the class of symmetric caching schemes in this paper. Consider a demand $\textbf{\textit{d}}$, where the user $U_l$ requires the file $W_{d_{l}}$,
	 \begin{equation}
	 \textbf{\textit{d}}=(W_{d_{1}},\dots, W_{d_{K}})
	 \end{equation}
	 Let $\pi(.)$ be a permutation operation defined over the set $\{1,\dots, K\}$ and $\pi^{-1}(.)$ be its inverse. Now consider another demand $\pi \textbf{\textit{d}}$, which is obtained by permuting the files requested by the users,
	 \begin{equation}
	 \pi \textbf{\textit{d}}=(W_{d_{\pi^{-1}(1)}},\dots,W_{d_{\pi^{-1}(K)}}).
	 \end{equation}
	 In the demand $\pi \textbf{\textit{d}}$, the user $U_{\pi(l)}$ requires the file $W_{d_{l}}$. In response to the demand $\pi \textbf{\textit{d}}$, the server broadcasts a set of packets $X_{\pi\textbf{\textit{d}}}$. For a symmetric caching scheme, we have \cite{tian2016symmetry}
	 \begin{equation}
	 \label{symmetry}
	 H(W_{d_{l}},Z_{\pi(l)},X_{\pi\textbf{\textit{d}}})=H(W_{d_{l}},Z_{l},X_{\textbf{\textit{d}}})
	 \end{equation}
	 \begin{equation}
	 \label{symmetry_file}
	 H(W_{n},Z_{l})=H(W_{p},Z_{l})
	 \end{equation}
	 where $n,p\in \{1,\dots, N\}$.

	Consider the demands where each of the $N$ files is required by at least one user (and hence $N\leq K$). The set of all such demands is denoted by $\textbf{\textit{D}}$ and the corresponding rate is denoted by $R(M)$, where 
	\begin{equation}
	R(M)=\max\{R_{\textbf{\textit{d}}}(M)\mid \textbf{\textit{d}}\in \textbf{\textit{D}}\}.
	\end{equation}
	For the $(N,K)$ cache network with cache size $M$, the memory rate pair $(M,R)$ is said to be achievable if there is a scheme with $R(M)\leq R$. For a such a  scheme, we have
	\begin{IEEEeqnarray}{rl}
	\label{M}
	H(Z_{l})\leq& M\\
	\label{R set}
	H(X_{\textbf{\textit{d}}})\leq& R\\
	\label{I.1}
	H(Z_{l},X_{\textbf{\textit{d}}})=&H(W_{d_{l}},Z_{l},X_{\textbf{\textit{d}}}),\\
	\label{I.2}
	H(W_{1},\dots,W_{N},Z_{l},X_{\textbf{\textit{d}}})=&H(W_{1},\dots,W_{N}),
	\end{IEEEeqnarray}
	where (\ref{M}) follows from the fact that size of each cache is $M$, (\ref{R set}) follows from the fact that  for any demand in $\textbf{\textit{D}}$ the size of $X_{\textbf{\textit{d}}}$ is  at most $R(M)\leq R$,  (\ref{I.1})  follows from the fact that the file $W_{\textbf{\textit{d}}_{l}}$ can be computed from $X_{\textbf{\textit{d}}}$ and $Z_{l}$ by the user $U_{l}$, and (\ref{I.2}) follows from the fact that  $Z_{l}$ and $X_{\textbf{\textit{d}}}$ are functions of files $\{W_{1},\dots, W_{N}\}$. For a given cache size $M$, the smallest $R$ such that $(M, R)$ is achievable is called the exact rate memory tradeoff denoted by
	\begin{equation}
	R^{*}(M)=\min \{R: (M, R)\text{ is achievable}\}
	\end{equation}

	Maddah-Ali and Niesen in \cite{maddah2014fundamental} proposed a coding scheme with an uncoded placement phase and a coded delivery phase for the demands in $\textbf{\textit{D}}$, and demonstrated that the rate achieved by the proposed scheme is within a multiplicative gap of 12 from the optimal rate using cut set arguments.  Several improvements to the scheme proposed in \cite{maddah2014fundamental} were presented in \cite{amiri2016coded, wan2016caching,wan2020index,yu2016exact}, and in a surprising result, Yu et al.  obtained the complete  characterization of the exact rate memory tradeoff when the placement phase is restricted to be uncoded \cite{yu2016exact}. In the general case, when coding is permitted in both the placement and the delivery phases, several attempts to reduce the gap where presented by deriving new lower bounds \cite{ ghasemi2017improved, ajaykrishnan2015critical, sengupta2015improved,  wang2018improved, yu2017characterizing}, and proposing new coding schemes  \cite{chen2014fundamental,amiri2016fundamental,gomez2018fundamental,tian2016caching,sahraei2016k, vijith2019towards,kp2019fundamental,Shao_Gomez-Vukardebo_Zhang_Tian_2019,shao2020fundamental}. The performance obtained by these schemes is summarised in TABLE I.
	\begin{table}[t]
		\centering
		\setlength{\tabcolsep}{2pt}
		\begin{tabular}{|c|c|c|c|}
			\hline 
			{Caching Scheme} & {Cache Size, $(M)$} &  Rate Memory Tradeoff& Condition  \\
			\hline&&&\\[-1.4em]   
			Chen et al. \cite{chen2014fundamental}  &   $\left[0,\frac{1}{K}\right]$  & $R^{*}(M)=N-NM$ & $N\leq K$  \\[0.5em]

			\hline&&&\\[-1.4em]    G{\'o}mez-Vilardeb{\'o} \cite{gomez2018fundamental}  & $\left[\frac{1}{N},\frac{1}{(N-1)}\right]$  &$R^{*}(M)=\dfrac{N^{2}-1}{N}-(N-1)M$& $K=N$\\[0.5em]

			\hline&&&\\[-1.4em]    $\begin{array}{c}
			\text{Maddah-Ali}\\ \text{and Niesen \cite{maddah2014fundamental} }
			\end{array}$   &$\left[\frac{N}{K}(K-1),N\right]$, & $R^{*}(M)=1-\frac{1}{N}M$ & 
			-\\[0.5em]
			
			\hline&&&\\[-1.4em]    Yu et al. \cite{yu2016exact}  & $[0,N]$  & $\begin{array}{c}
			R(M)=R_{r}+(R_{r}-R_{r+1})\left(r-\frac{N}{K}M\right)\\
			\text{where }R_r=\frac{{}^{K}\!C_{r+1}-{}^{K-N}\!C_{r+1}}{{}^{K}\!C_{r}}\\
			\text{and $r\in \{1,\dots, K\}$}
			\end{array}$  &$\begin{array}{c}
			\text{Optimal among}\\ \text{uncoded prefetching}\\ \text{schemes}
			\end{array}$\\[0.5em]
			
			\hline&&&\\[-1.4em]    Vijith et al. \cite{vijith2019towards}, \cite{kp2019fundamental} & $\left[N-1-\frac{1}{N},N-1\right]$&$R^{*}(M)=\frac{N+1}{N}-\frac{1}{N-1}M$& $N=K$  \\ [0.5em]
			
			\hline&&&\\[-1.4em]
			\multirow{3}{*}{This paper}   & $\left[\frac{N}{K}\Big((K-2)+\frac{(K-2+1/N)}{(K-1)}\Big), \frac{N}{K}(K-1)\right]$ &$R^{*}(M)=\frac{(KN-1)}{K(N-1)}-\frac{1}{(N-1)}M$& $\big\lceil \frac{K+1}{2}\big\rceil \leq N \leq K $ 
			\\ [0.8em]
			\hhline{~---}&&&\\[-1.4em]
			& $\left[\frac{N}{K}(K-2),\frac{N}{K}(K-1)\right]$ & $R^{*}(M)=\frac{K^{2}+K-2}{K(K-1)}-\frac{(K+1)}{N(K-1)}M$&$\begin{array}{c}
			1\leq N\leq \big\lceil \frac{K+1}{2}\big\rceil
			\end{array}$ 
			\\ [0.5em]
			\hline 
		\end{tabular}
		\caption{Summary of previous work in coded caching}
		\label{Previous works}
	\end{table}

	  The contributions of this paper are as follows:
	 \begin{itemize}
	 		\item We propose a new coding scheme for the $(N,K)$ cache network to achieve the memory rate pair $\Big(\frac{N}{K}\big(K-2+\frac{(K-2+1/N)}{(K-1)}\big),\frac{1}{K-1}\Big)$.  
	 		\item For the case $\big\lceil\frac{K+1}{2}\big\rceil\leq N \leq K$, we derive a matching lower bound and obtain a characterization of the exact rate memory tradeoff when $M\geq \frac{N}{K}\Big(K-2+\frac{(K-2+1/N)}{(K-1)}\Big)$. 
	 		\item For the case $1\leq N\leq \big\lceil\frac{K+1}{2}\big\rceil$, we derive a new lower bound to
	 		match the scheme proposed by Yu et al. \cite{yu2016exact} and obtain a characterization of the exact rate memory tradeoff when $M\geq \frac{N}{K}(K-2)$.
	 \end{itemize}
 	Throughout this paper we use $[L]$ to represent the set $\{1,2,\dots,L\}$,  and $W_{[L]}$ to represent the set $\{W_{1}, W_{2},\dots, W_{L}\}$.

	\section{Example networks}
	As a prelude to the results presented in Section III and IV, we consider two example networks.
	\subsection{The (3,4) Cache Network}
	\label{example(3,4)}
	\noindent Here, users $\{U_{1},U_{2},U_{3},U_{4}\}$ are connected to a server with files $\{A,B,C\}$ (each of size $F$ bits). Each user $U_{k}$ has a cache $Z_{k}$ of size $MF$ bits. We now describe a symmetric caching scheme for the case  $M=\frac{25}{12}$. During the placement phase, every file is split into 12 disjoint subfiles, each of size $\frac{1}{12}F$ bits. The subfiles are:
	\begin{table}[ht]
		\centering
		\begin{tabular}{|c|c|}
		\hline
		File& Subfiles\\
		\hline
		$A$&$A^{12},$ $A^{21},$ $A^{13},$ $A^{31},$ $A^{14},$ $A^{41},$ $ 
		A^{23},$ $A^{32},$ $A^{24},$ $A^{42},$ $A^{34},$ $A^{43}$
		\\
		\hline
		$B$&$B^{12},$ $B^{21},$ $B^{13},$ $B^{31},$ $B^{14},$ $B^{41},$ $ 
		B^{23},$ $B^{32},$ $B^{24},$ $B^{42},$ $B^{34},$ $B^{43}$
		\\
		\hline
		$C$&$C^{12},$ $C^{21},$ $C^{13},$ $C^{31},$ $C^{14},$ $C^{41},$ $ 
		C^{23},$ $C^{32},$ $C^{24},$ $C^{42},$ $C^{34},C^{43}$\\
		\hline
	\end{tabular}
	\end{table}

	\noindent The server places  18 uncoded packets (stage 1) and 7 coded packets (stage 2) in each user's cache as shown in TABLE \ref{table:cache content}. Each of these packets are of size $\frac{1}{12}F$ bits and they together occupy $\frac{25}{12}F$ bits.
		\begin{table}[ht]
		\centering
		\begin{tabular}{|c|c|c|c|c|c|c|c|c|c|}
			\hline
			Cache&\multicolumn{6}{c|}{Stage 1}&\multicolumn{3}{c|}{Stage 2}\\
			\hline
			\multirow{3}{*}{$Z_{1}$}&$A^{23}$&$A^{24}$&$A^{32}$&$A^{34}$&$A^{42}$&$A^{43}$&$A^{12}- A^{13}$&$A^{12}- A^{14}$&\\
			\hhline{~---------} 
			&$B^{23}$&$B^{24}$&$B^{32}$&$B^{34}$&$B^{42}$&$B^{43}$&$B^{12}- B^{13}$&$B^{12}- B^{14}$&\\
			\hhline{~---------}
			&$C^{23}$&$C^{24}$&$C^{32}$&$C^{34}$&$C^{42}$&$C^{43}$&$C^{12}- C^{13}$&$C^{12}- C^{14}$& $A^{12}+ B^{12}+ C^{12}$\\
			\hline 
			\multirow{3}{*}{$Z_{2}$}&$A^{13}$&$A^{14}$&$A^{31}$&$A^{34}$&$A^{41}$&$A^{43}$&$A^{23}- A^{24}$&$A^{23}- A^{21}$&\\
			\hhline{~---------}
			&$B^{13}$&$B^{14}$&$B^{31}$&$B^{34}$&$B^{41}$&$B^{43}$&$B^{23}- B^{24}$&$B^{23}- B^{21}$&\\
			\hhline{~---------} 
			&$C^{13}$&$C^{14}$&$C^{31}$&$C^{34}$&$C^{41}$&$C^{43}$&$C^{23}- C^{24}$&$C^{23}- C^{21}$& $A^{23}+ B^{23}+ C^{23}$\\
			\hline
			\multirow{3}{*}{$Z_{3}$}&$A^{21}$&$A^{24}$&$A^{12}$&$A^{14}$&$A^{42}$&$A^{41}$&$A^{34}- A^{31}$&$A^{34}- A^{32}$&\\
			\hhline{~---------} 
			&$B^{21}$&$B^{24}$&$B^{12}$&$B^{14}$&$B^{42}$&$B^{41}$&$B^{34}- B^{31}$&$B^{34}- B^{32}$&\\
			\hhline{~---------}
			&$C^{21}$&$C^{24}$&$C^{12}$&$C^{14}$&$C^{42}$&$C^{41}$&$C^{34}- C^{31}$&$C^{34}- C^{32}$& $A^{34}+ B^{34}+ C^{34}$\\
			\hline
			\multirow{3}{*}{$Z_{4}$}&$A^{23}$&$A^{21}$&$A^{32}$&$A^{31}$&$A^{12}$&$A^{13}$&$A^{41}- A^{42}$&$A^{41}- A^{43}$&\\
			\hhline{~---------}
			&$B^{23}$&$B^{21}$&$B^{32}$&$B^{31}$&$B^{12}$&$B^{13}$&$B^{41}- B^{42}$&$B^{41}- B^{43}$&\\
			\hhline{~---------}
			&$C^{23}$&$C^{21}$&$C^{32}$&$C^{31}$&$C^{12}$&$C^{13}$&$C^{41}- C^{42}$&$C^{41}- C^{43}$& $A^{41}+ B^{41}+ C^{41}$\\
			\hline
			
		\end{tabular}
		\caption{Cache contents placed in stage 1 and stage 2}
		\label{table:cache content}
	\end{table}

	\noindent To understand how the delivery phase works, consider a demand $\textbf{\textit{d}}=(P,P,Q,R)$ where $P$, $Q$ and $R$  are distinct files in $\{A,B,C\}$. In response to this demand, the server broadcasts a set of packets
	\begin{equation*}
	X_{\textbf{\textit{d}}}=\left\{\begin{array}{c}
	Q^{13}+R^{14}-P^{12}\\ Q^{23}+R^{24}-P^{21}\\
	R^{34}+\frac{1}{2}P^{31}+\frac{1}{2}P^{32} \\ Q^{43}+\frac{1}{2}P^{41}+\frac{1}{2}P^{42}
	\end{array}\right\}.
	\end{equation*}
	\noindent As $X_{\textbf{\textit{d}}}$ has four packets, of size $\frac{1}{12}F$ bits each, the load experienced by the shared link is $RF=\frac{1}{3}F$ bits and thus the rate is $R=\frac{1}{3}$.
	
	Let us consider $U_1$ to understand how the requested file is obtained from $X_{\textbf{\textit{d}}}$ and $Z_{1}$. Note that the user has subfiles $P^{23},P^{32},P^{24},P^{42},P^{34}
	\text{ and }P^{43}$ in $Z_{1}$ and require subfiles $P^{12}$, $P^{21}$, $P^{13}$, $P^{31}$, $P^{14}$ and $P^{41}$ to compute the requested file $P$. The user can compute subfiles $P^{21}$, $P^{31}$, $P^{41}$ and $P^{12}$  by combining received packets and cached packets as shown below:

	\begin{table}[!ht]
		\centering
		\begin{tabular}{|c|c|c|}
			\hline
			Received Packet& Cached Packets & Computed Subfile\\
			\hline
			$Q^{23}+R^{24}-P^{21}$&$Q^{23}$, $R^{24}$&$P^{21}$\\
			\hline
			$R^{34}-\frac{1}{2}P^{31}-\frac{1}{2}P^{32}$&$R^{34}$, $P^{32}$& $P^{31}$\\
			\hline
			$Q^{43}-\frac{1}{2}P^{41}-\frac{1}{2}P^{42}$&$Q^{43}$, $P^{42}$&$P^{41}$\\
			\hline
			$Q^{13}+R^{14}-P^{12}$&$Q^{12}-Q^{13}$, $R^{12}-R^{14}$, $P^{12}+Q^{12}+R^{12}$& $P^{12}$\\
			\hline
		\end{tabular}
	\end{table}

	 \noindent Combining the subfile $P^{12}$ with cached packets $P^{12}-P^{13}$ and $P^{12}-P^{14}$, $U_{1}$ obtains subfiles $P^{13}$ and $P^{14}$. The other users can proceed in similar fashion. We summarise as:
	\begin{lemma}
		\label{eg_case1_rate}
		The memory rate pair $\big(\frac{25}{12},\frac{1}{3}\big)$ is achievable by symmetric caching schemes for the $(3,4)$ cache network.
	\end{lemma}
	The caching scheme proposed in \cite{maddah2014fundamental} achieves the memory rate pair $\big(\frac{9}{4},\frac{1}{4}\big)$, and by memory sharing between that scheme and the proposed scheme, we can  achieve all memory rate pairs $\big(M,\frac{11}{8}-\frac{1}{2}M\big)$, where $M\in \big[\frac{25}{12},\frac{9}{4}\big]$. We obtain a matching lower bound in the following lemma:

	\begin{lemma}
		\label{eg_case1_converse}
		For the $(3,4)$ cache network, achievable memory rate pairs $(M,R)$ must satisfy the constraint
		\begin{equation*}
		4M+8R\geq 11.
		\end{equation*}
	\end{lemma}
	\begin{proof}
We have,
		\begin{IEEEeqnarray*}{rl}
		4M+8R\overset{(a)}{\geq}& 2H(Z_{1})+H(Z_{2})+H(Z_{3})+2H(X_{(A,B,C,A)})+2H(X_{(B,C,A,A)})+3H(X_{(A,A,B,C)})\\&+H(X_{(C,A,A,B)})\\
				\overset{(b)}{\geq}& H(Z_{1},X_{(A,B,C,A)},X_{(B,C,A,A)})+H(Z_{1},X_{(B,C,A,A)},X_{(A,A,B,C)})+H(Z_{2},X_{(A,B,C,A)},X_{(A,A,B,C)})\\&+H(Z_{3},X_{(C,A,A,B)},X_{(A,A,B,C)})\\	
		\overset{(c)}{=}& H(A,B,Z_{1},X_{(A,B,C,A)},X_{(B,C,A,A)})+H(A,B,Z_{1},X_{(B,C,A,A)},X_{(A,A,B,C)})\\&+H(A,B,Z_{2},X_{(A,B,C,A)},X_{(A,A,B,C)})+H(A,B,Z_{3},X_{(C,A,A,B)},X_{(A,A,B,C)})	\\
		\overset{(b)}{\geq}& H(A,B,Z_{1},X_{(B,C,A,A)})+H(A,B,Z_{1},X_{(A,B,C,A)},X_{(B,C,A,A)},X_{(A,A,B,C)})\\&+H(A,B,Z_{2},X_{(A,B,C,A)},X_{(A,A,B,C)})+H(A,B,Z_{3},X_{(C,A,A,B)},X_{(A,A,B,C)})\\
		\overset{}{\geq}& H(A,B,X_{(A,B,C,A)},X_{(B,C,A,A)},X_{(A,A,B,C)})+H(A,B,Z_{2},X_{(A,B,C,A)},X_{(A,A,B,C)})\\&+H(A,B,Z_{3},X_{(C,A,A,B)},X_{(A,A,B,C)})+H(A,B,Z_{1},X_{(B,C,A,A)})\\
		\overset{(b)}{\geq}& H(A,B,Z_{2},X_{(A,B,C,A)},X_{(B,C,A,A)},X_{(A,A,B,C)})+H(A,B,X_{(A,B,C,A)},X_{(A,A,B,C)})\\&+H(A,B,Z_{3},X_{(C,A,A,B)},X_{(A,A,B,C)})+H(A,B,Z_{1},X_{(B,C,A,A)})\\
		\overset{(c)}{=}& H(A,B,C,Z_{2},X_{(A,B,C,A)},X_{(B,C,A,A)},X_{(A,A,B,C)})+H(A,B,X_{(A,B,C,A)},X_{(A,A,B,C)})\\&+H(A,B,Z_{3},X_{(C,A,A,B)},X_{(A,A,B,C)})+H(A,B,Z_{1},X_{(B,C,A,A)})\\
		\overset{(d)}{=}& H(A,B,C)+H(A,B,X_{(A,B,C,A)},X_{(A,A,B,C)})+H(A,B,Z_{3},X_{(C,A,A,B)},X_{(A,A,B,C)})\\&+H(A,B,Z_{1},X_{(B,C,A,A)})\\
		\overset{(b)}{\geq}& H(A,B,C)+H(A,B,Z_{3},X_{(C,A,A,B)},X_{(A,B,C,A)},X_{(A,A,B,C)})+H(A,B,X_{(A,A,B,C)})\\&+H(A,B,Z_{1},X_{(B,C,A,A)})\\	
		\overset{(c)}{=}& H(A,B,C)+H(A,B,C,Z_{3},X_{(C,A,A,B)},X_{(A,B,C,A)},X_{(A,A,B,C)})+H(A,B,X_{(A,A,B,C)})\\&+H(A,B,Z_{1},X_{(B,C,A,A)})\\
		\overset{(d)}{=}& 2H(A,B,C)+H(A,B,X_{(A,A,B,C)})+H(A,B,Z_{1},X_{(B,C,A,A)})\\
		\overset{}{\geq}& 2H(A,B,C)+H(A,B,X_{(A,A,B,C)})+H(A,B,Z_{1})\\
		\overset{(e)}{=}& 2H(A,B,C)+H(A,B,X_{(A,A,B,C)})+H(A,B,Z_{4})\\
		\overset{(b)}{\geq}& H(A,B,Z_{4},X_{(A,A,B,C)})+H(A,B)+2H(A,B,C)\\
		\overset{(c)}{=}& H(A,B,C,Z_{4},X_{(A,A,B,C)})+H(A,B)+2H(A,B,C)\\
		\overset{(d)}{=}& H(A,B,C)+H(A,B)+2H(A,B,C)
		\overset{}{\geq}11
		\end{IEEEeqnarray*}	
		where
		
		\hspace{-4mm}\begin{tabular}{rl}
		     $(a)$& follows from (\ref{M}) and (\ref{R set}),\\ 
		     $(b)$& follows from the submodularity property of entropy,\\
		     $(c)$& follows from (\ref{I.1}),\\
		     $(d)$& follows from (\ref{I.2}),\\
		     $(e)$& follows from (\ref{symmetry}).
		\end{tabular}
	
	\end{proof}

The above observations improve upon the previous results from \cite{maddah2014fundamental,yu2016exact} and is summarised in TABLE \ref{Rate achieved for the $(3,4)$ cache network} and Fig. \ref{fig:(3,4)bound}.
	\begin{table}[ht]
		\centering
		\begin{tabular}{|c|c|c|c|c|}
			\hline
			Memory &Rate \cite{maddah2014fundamental,yu2016exact}&Lower Bound\cite{maddah2014fundamental,yu2016exact}&New Rate& New Lower Bound\\
			\hline&&&&\\[-1.em]
			$\dfrac{25}{12}\leq M\leq \dfrac{9}{4}$&$\dfrac{3}{2}-\dfrac{5}{9}M$ &$R\geq1-\dfrac{1}{3}M$&$\dfrac{11}{8}-\dfrac{1}{2}M$&$R\geq\dfrac{11}{8}-\dfrac{1}{2}M$\\[0.4em]
			\hline
		\end{tabular}
	\caption{Rate memory tradeoff for the $(3,4)$ cache network}
	\label{Rate achieved for the $(3,4)$ cache network}
	\end{table}

\begin{figure}[ht]
	\centering
	\begin{tikzpicture}[line cap=round,line join=round,x=3.55cm,y=1.7cm,
    spy/.style={%
        draw,green,
        line width=1pt,
        rectangle,inner sep=0pt,
    },
]
   
    \def\spyviewersize{2.4cm}
    \def\spyonclipreduce{0.8pt}

    \def\spyfactorI{50}
    \coordinate (spy-on 1) at (25/12,1/3);
    \coordinate (spy-in 1) at (2.3,1.5);

    \def\pic{\coordinate (O) at (0,0);
       \draw [ultra thin,step=3,black] (0,0) grid (3,3);

   \foreach \x in {0,1/4,3/4,6/4,9/4,3}
   \draw[shift={(\x,0)},color=black,thin] (0pt,1pt) -- (0pt,-1pt)
                                   node[below] {\footnotesize $\x$};

      \foreach \y in {0,1,2,3}
      \draw[shift={(0,\y)},color=black,thin] (1pt,0pt) -- (-1pt,0pt)
                                    node[left] {\footnotesize $\y$};
  \draw[color=black] (6cm,-18pt) node[left] { Cache size $M$};
  \draw[color=black] (-18pt,3.cm) node[left,rotate=90] { Rate $R$};

   \draw[smooth,blue!70,mark=otimes,samples=1000,domain=0.0:2.2,mark = $\otimes$,line width=2pt]
      {(0,3)--(1/4,9/4)  node[mark size=2.5pt,line width=2pt,label={right:$(\frac{1}{4},\frac{9}{4})$}]{$\pgfuseplotmark{square*}$} (9/4,1/4)node[mark size=2.5pt,line width=2pt,label={above right:$(\frac{9}{4},\frac{1}{4})$}]{$\pgfuseplotmark{square*}$}--(3,0)};

      \draw[smooth,red!60!black,samples=1000,domain=0.0:2.2,mark size=2pt,mark =otimes*,line width=2pt]
      {(25/12,1/3)node[mark size=4pt,line width=2pt,label={above:$(\frac{25}{12},\frac{1}{3})$}]{$\pgfuseplotmark{triangle*}$}--(9/4,1/4)};
      
    }

    \pic

\draw [gray!50!white,line width=1pt,fill=white] (0.8,2.2)rectangle (3,3);
\begin{scope}[shift={(0.6,2.4)}] 
\draw [smooth,samples=1000,domain=0.0:2.2,red!60!black,mark=otimes,line width=2pt] 
{(0.25,0) --node [mark size=4pt,line width=0.1pt]{$\pgfuseplotmark{triangle*}$} (0.5,0)}
node[right]{New Rate Memory Tradeoff};

\draw [yshift=0.75\baselineskip,smooth,blue!70,samples=1000,domain=0.0:2.2,mark=otimes,line width=2pt] 
{(0.25,0) --node [mark size=2.5pt,line width=0.1pt]{$\pgfuseplotmark{square*}$} (0.5,0)}
node[right]{Rate Memory Tradeoff \cite{maddah2014fundamental,yu2016exact,chen2014fundamental}};
	\end{scope}

\end{tikzpicture}
	\caption{Exact rate memory tradeoff for the $(3,4)$ cache network}
	\label{fig:(3,4)bound}
\end{figure}
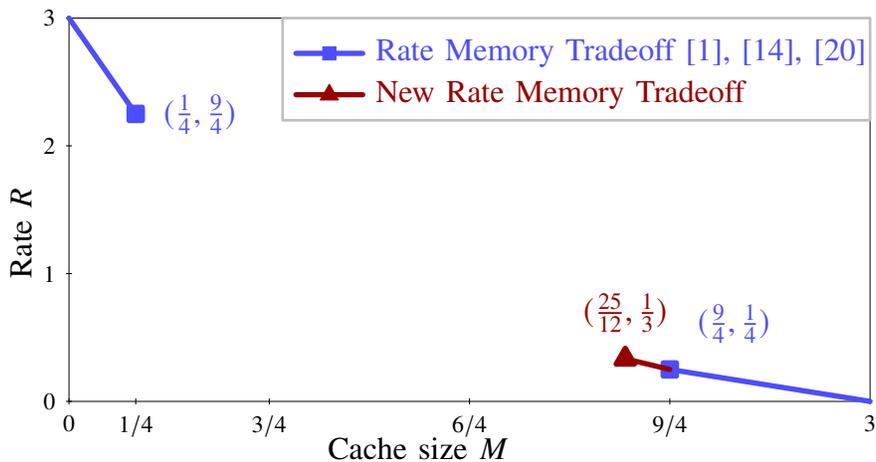

		\subsection{The (2,4) Cache Network}
		\noindent Here, users $\{U_{1},U_{2},U_{3},U_{4}\}$ are connected to a server with files $\{A,B\}$ (each of size $F$ bits). Each user $U_{k}$ has cache $Z_{k}$ of size  $MF$ bits. The caching scheme proposed in \cite{yu2016exact} can achieve all memory rate pairs $\Big(M,\frac{3}{2}-\frac{5}{6}M\Big)$, where $M\in \big[1,\frac{3}{2}\big]$. We obtain a matching lower bound in the following lemma:
		
		\begin{lemma}
			\label{eg_case2_converse}
			For the $(2,4)$ cache network, achievable memory rate pairs $(M,R)$ must satisfy the constraint
			\begin{equation*}
			5M+6R\geq 9.
			\end{equation*}
		\end{lemma}
		\begin{proof}
			We have,
			\begin{IEEEeqnarray*}{rl}
			5M+&6R\overset{(a)}{\geq} 3H(Z_{1})+2H(Z_{2})+H(Z_{4})+3H(X_{(A,A,B,A)}) +2H(X_{(A,A,A,B)})+H(X_{(A,B,A,A)})-H(Z_{1})\\
					\overset{(b)}{\geq}& H(Z_{1},X_{(A,A,B,A)})+H(Z_{1},X_{(A,A,A,B)})+H(Z_{1},X_{(A,B,A,A)})+H(Z_{2},X_{(A,A,B,A)})+H(Z_{2},X_{(A,A,A,B)})\\&+H(Z_{4},X_{(A,A,B,A)})-H(Z_{1})\\
			\overset{(c)}{=}& H(A,Z_{1},X_{(A,A,B,A)})+H(A,Z_{1},X_{(A,A,A,B)})+H(A,Z_{1},X_{(A,B,A,A)})+H(A,Z_{2},X_{(A,A,B,A)})\\&+H(A,Z_{2},X_{(A,A,A,B)})+H(A,Z_{4},X_{(A,A,B,A)})-H(Z_{1})\\
			\overset{(b)}{\geq}& H(A,Z_{1},X_{(A,A,B,A)},X_{(A,A,A,B)},X_{(A,B,A,A)})+2H(A,Z_{1})+H(A,Z_{2},X_{(A,A,B,A)},X_{(A,A,A,B)})\\&+H(A,Z_{2})+H(A,Z_{4},X_{(A,A,B,A)})-H(Z_{1})\\
			\overset{(d)}{=}& H(A,Z_{1},X_{(A,A,B,A)},X_{(A,A,A,B)},X_{(A,B,A,A)})+H(A,Z_{2},X_{(A,A,B,A)},X_{(A,A,A,B)})\\&+H(A,Z_{4},X_{(A,A,B,A)})+H(A,Z_{2})+H(A,Z_{1})+H(B,Z_{1})-H(Z_{1})\\
			\overset{(e)}{=}& H(A,Z_{1},X_{(A,A,B,A)},X_{(A,A,A,B)},X_{(A,B,A,A)})+H(A,Z_{2},X_{(A,A,B,A)},X_{(A,A,A,B)})\\&+H(A,Z_{4},X_{(A,A,B,A)})+H(A,Z_{3})+H(A,Z_{1})+H(B,Z_{1})-H(Z_{1})\\
			\overset{(b)}{\geq}& H(A,Z_{1},X_{(A,A,B,A)},X_{(A,A,A,B)},X_{(A,B,A,A)})+H(A,Z_{2},X_{(A,A,B,A)},X_{(A,A,A,B)})+H(A,Z_{4},X_{(A,A,B,A)})\\&+H(A,Z_{3})+H(A,B,Z_{1})+H(Z_{1})-H(Z_{1})\\
			\overset{(f)}{=}& H(A,Z_{1},X_{(A,A,B,A)},X_{(A,A,A,B)},X_{(A,B,A,A)})+H(A,Z_{2},X_{(A,A,B,A)},X_{(A,A,A,B)})+H(A,Z_{4},X_{(A,A,B,A)})\\&+H(A,Z_{3})+H(A,B)\\
			\overset{(b)}{\geq}& H(A,Z_{1},Z_{2},X_{(A,A,B,A)},X_{(A,A,A,B)},X_{(A,B,A,A)})+H(A,X_{(A,A,B,A)},X_{(A,A,A,B)})+H(A,Z_{4},X_{(A,A,B,A)})\\&+H(A,Z_{3})+H(A,B)\\
			\overset{(c)}{=}& H(A,B,Z_{1},Z_{2},X_{(A,A,B,A)},X_{(A,A,A,B)},X_{(A,B,A,A)})+H(A,X_{(A,A,B,A)},X_{(A,A,A,B)})+H(A,Z_{4},X_{(A,A,B,A)})\\&+H(A,Z_{3})+H(A,B)\\
			\overset{(f)}{=}& 2H(A,B)+H(A,X_{(A,A,B,A)},X_{(A,A,A,B)}) +H(A,Z_{4},X_{(A,A,B,A)})+H(A,Z_{3})\\
			\overset{(b)}{\geq}& 2H(A,B)+H(A,Z_{4},X_{(A,A,B,A)},X_{(A,A,A,B)}) +H(A,X_{(A,A,B,A)})+H(A,Z_{3})\\
			\overset{(c)}{=}& 2H(A,B)+H(A,B,Z_{4},X_{(A,A,B,A)},X_{(A,A,A,B)}) +H(A,X_{(A,A,B,A)})+H(A,Z_{3})\\
			\overset{(f)}{=}& 3H(A,B)+H(A,X_{(A,A,B,A)})+H(A,Z_{3})\\
			\overset{(b)}{\geq}& 3H(A,B)+H(A,Z_{3},X_{(A,A,B,A)})+H(A)\\
			\overset{(c)}{=}& 3H(A,B)+H(A,B,Z_{3},X_{(A,A,B,A)})+H(A)\\
			\overset{(f)}{=}& 4H(A,B)+H(A)\geq9,
			\end{IEEEeqnarray*}
			where
			
			\hspace{-4mm}\begin{tabular}{rl}
			     $(a)$& follows from (\ref{M}) and (\ref{R set}),\\
			     $(b)$& follows from the submodularity property of entropy,\\
			     $(c)$& follows from (\ref{I.1}),\\
			     $(d)$& follows from (\ref{symmetry_file}),\\
			     $(e)$& follows from (\ref{symmetry}),\\
			     $(f)$& follows from (\ref{I.2})
			\end{tabular}

		\end{proof}
The above observations improve upon the previous results from \cite{maddah2014fundamental,yu2016exact} and is summarised in TABLE \ref{Rate achieved for the $(2,4)$ cache network} and Fig. \ref{fig:(2,4)bound}.
	\begin{table}[!ht]
		\centering
		\begin{tabular}{|c|c|c|c|}
			\hline
			Memory&Rate \cite{maddah2014fundamental,yu2016exact}& Lower Bound \cite{maddah2014fundamental,yu2016exact}& New Lower Bound\\
			\hline&&&\\[-1em]
			$1\leq M\leq\dfrac{3}{2}$&$\dfrac{3}{2}-\dfrac{5}{6}M$&$R\geq 1-\dfrac{1}{2}M$&$R\geq\dfrac{3}{2}-\dfrac{5}{6}M$\\[0.6em]
			\hline
		\end{tabular}
	\caption{Rate memory tradeoff for the $(2,4)$ cache network}
	\label{Rate achieved for the $(2,4)$ cache network}
	\end{table}
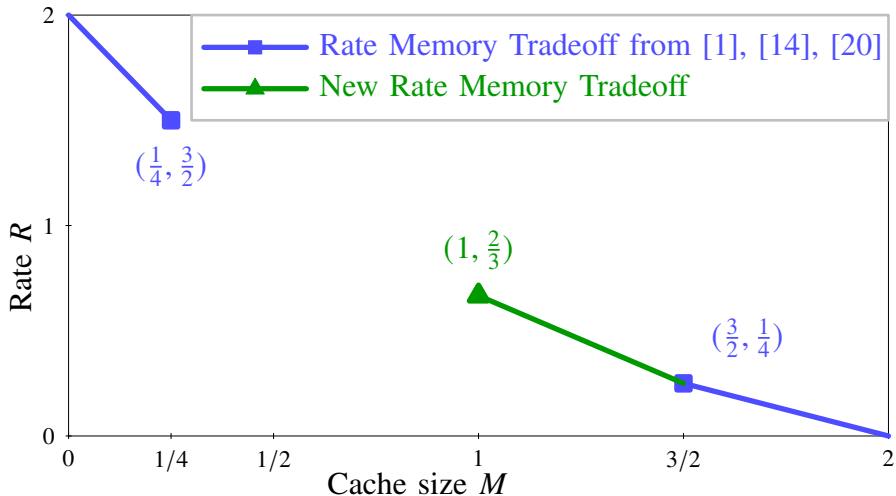
\begin{figure}[ht]
	\centering
	\begin{tikzpicture}[line cap=round,line join=round,x=5.45cm,y=2.8cm,
    spy/.style={%
        draw,green,
        line width=1pt,
        rectangle,inner sep=0pt,
    },
]

    \def\spyviewersize{2.4cm}

    \def\spyonclipreduce{0.8pt}

    \def\spyfactorI{50}
    \coordinate (spy-on 1) at (25/12,1/3);
    \coordinate (spy-in 1) at (2.3,1.5);

    \def\pic{\coordinate (O) at (0,0);
       \draw [ultra thin,step=2,black] (0,0) grid (2,2);

   \foreach \x in {0,1/4,1/2,1,3/2,2}
   \draw[shift={(\x,0)},color=black,thin] (0pt,1pt) -- (0pt,-1pt)
                                   node[below] {\footnotesize $\x$};
      \foreach \y in {0,1,2}
      \draw[shift={(0,\y)},color=black,thin] (1pt,0pt) -- (-1pt,0pt)
                                    node[left] {\footnotesize $\y$};
  \draw[color=black] (6cm,-18pt) node[left] { Cache size $M$};
  \draw[color=black] (-18pt,3.cm) node[left,rotate=90] { Rate $R$};

   \draw[smooth,blue!70,mark=otimes,samples=1000,domain=0.0:2.2,mark = $\otimes$,line width=2pt]
      {(0,2)--(1/4,3/2)  node[mark size=2.5pt,line width=2pt,label={below:$(\frac{1}{4},\frac{3}{2})$}]{$\pgfuseplotmark{square*}$} (3/2,1/4)node[mark size=2.5pt,line width=2pt,label={above right:$(\frac{3}{2},\frac{1}{4})$}]{$\pgfuseplotmark{square*}$}--(2,0)};
            
      \draw[smooth,green!60!black,samples=1000,domain=0.0:2.2,mark size=2pt,mark =otimes*,line width=2pt]
      {(1,2/3)node[mark size=4pt,line width=2pt,label={above:$(1,\frac{2}{3})$}]{$\pgfuseplotmark{triangle*}$}--(3/2,1/4)};

    }

    \pic

\draw [gray!50!white,line width=1pt,fill=white] (0.3,1.5)rectangle (2,2);
\begin{scope}[shift={(0.08,1.65)}] 
		\draw [smooth,samples=1000,domain=0.0:2.2,green!60!black,mark=otimes,line width=2pt] 
	{(0.25,0) --node [mark size=4pt,line width=0.1pt]{$\pgfuseplotmark{triangle*}$} (0.5,0)}
	node[right]{New Rate Memory Tradeoff};

		\draw [yshift=0.75\baselineskip,smooth,blue!70,samples=1000,domain=0.0:2.2,mark=otimes,line width=2pt] 
				{(0.25,0) --node [mark size=2.5pt,line width=0.1pt]{$\pgfuseplotmark{square*}$} (0.5,0)}
				node[right]{Rate Memory Tradeoff from \cite{maddah2014fundamental,yu2016exact,chen2014fundamental}};

	\end{scope}

\end{tikzpicture}
	\caption{Exact rate memory tradeoff for the $(2,4)$ cache network}
	\label{fig:(2,4)bound}
\end{figure}

	\section[Case I]{Case I: $\big\lceil\frac{K+1}{2}\big\rceil \leq N\leq K$}
	In this section we propose a new symmetric caching scheme that achieves the memory rate pair
	\begin{equation}
	(M_{A},R_{A})=\bigg(\frac{N}{K}\Big(K-2+\frac{(K-2+1/N)}{(K-1)}\Big),\frac{1}{K-1}\bigg)
	\end{equation}
	for the $(N, K)$ cache network. This scheme can be seen as a generalization of the scheme presented for $(3,4)$ cache network in Section \ref{example(3,4)} and is an extension of the scheme we proposed in \cite{vijith2019towards,kp2019fundamental}. For $\big\lceil\frac{K+1}{2}\big\rceil\leq N\leq K$, we prove a matching lower bound to establish the exact rate memory tradeoff when $M\geq M_{A}$. Let $\textbf{I}$ denote the indicator function and let  $\textbf{S}_{k}$ denote the set $[K]\setminus\{k\}$. 
	\subsection{Placement and Delivery Phase}
	\noindent During the placement phase, the server splits every file into $2{\text{ } }{}^{K}\!C_{2}$ disjoint subfiles of size $\frac{1}{K(K-1)}F$ bits. Subfiles of the file $W_{n}$ are:
	\begin{equation*}
	\{W^{ij}_{n}: i,j\in [K] \text{ and } i\neq j\}.
	\end{equation*}
	The placement phase proceeds in two stages. In the first stage, the server copy subfiles $W_{n}^{ij}$ in user $U_k$'s cache, $Z_{k}$, if $k\notin \{i,j\}$. In the second stage,  functions of subfiles are computed and placed into each user's cache resulting in the cache of the $k^{th}$ user, $Z_{k}$, having the contents:
\begin{table}[ht]
\centering
\begin{tabular}{|c|c|c|c|}
	\hline &&&\\[-0.6em]
	Stages&Packets& Constraints &$\begin{array}{c}
	\text{Number}
	\end{array}$\\
	\hline
	Stage 1&$W_{n}^{ij}$& $\begin{array}{c}
	n\in [N] ,\\\text{$i,j\in \textbf{S}_{k}$ and $i\ne j$}
	\end{array}$ & $2N\text{ }{}^{K-1}\!C_{2}$ \\
	\hline&&&\\[-1.1em]
	\multirow{2}{*}{Stage 2}&$W_{n}^{k(k+1)}-W_{n}^{kj}$ & $\begin{array}{c}
	n\in [N]\text{ and $j\in \textbf{S}_{k}$}
	\end{array}$  & $N(K-2)$ \\
	\hhline{~---}&&&\\[-1.1em] 
	&$\sum_{n=1}^{N}W_{n}^{k(k+1)}$&  & 1 \\[0.3em]
	\hline 
\end{tabular} 
\end{table}

	\noindent It can be noted that subfiles $W_{n}^{kj}$, for $n\in [N]$ and $j\in \textbf{S}_{k}$, are contained in user $U_k$'s cache in coded form. The total number of packets, each of size  $\frac{1}{K(K-1)}F$ bits, placed in each user's cache is,
	\begin{equation*}
	2N\text{ }{}^{K-1}\!C_{2}+N(K-2)+1=NK(K-2)+1
	\end{equation*}
	utilising the entire cache of size $M_{A}F$ bits.

	In the delivery phase, let the server receive a demand $\textbf{\textit{d}}$. Let $N_{k}^{i}$ represent the number of users in the set $\textbf{S}_{k}$  requesting the file $W_{d_{i}}$. For each $k\in [K]$, the server constructs a packet,
	\begin{equation}
	X_{\textbf{\textit{d}}}^{k}=\sum_{s\in \textbf{S}_{k}}\bigg(\frac{\alpha_{k}^{s}}{N_{k}^{s}}\bigg)W_{d_{s}}^{ks}
	\end{equation}
	where
	\begin{equation}
	\alpha_{k}^{s}=1-2 \mathbf{I}\{W_{d_{k}}=W_{d_{s}}\}.
	\end{equation}
	The set of packets broadcast by the server in response to the demand $\textbf{\textit{d}}$ is,
	\begin{equation}
	X_{\textbf{\textit{d}}}=\{X_{\textbf{\textit{d}}}^{1},\dots,X_{\textbf{\textit{d}}}^{K}\}.
	\end{equation}
	Thus, $K$ packets, each of size $\frac{1}{K(K-1)}F$ bits, are transmitted and the rate corresponding to the demand $\textbf{\textit{d}}$ is,
	\begin{equation}
	R=\frac{1}{K-1}.
	\end{equation}

	\subsection{File Recovery by Users}
	\noindent To understand how the requested files are recovered  by the  users, let us consider user $U_k$ who needs to recover the file $W_{d_{k}}$ from its cache contents $Z_{k}$ and the received packets $X_{\textbf{\textit{d}}}$. Subfiles $W_{d_{k}}^{ij}$, for $i,j\in \textbf{S}_{k}$, are available in $Z_{k}$. To reconstruct the file $W_{d_{k}}$, the user needs to compute subfiles $W_{d_{k}}^{jk}$, and $W_{d_{k}}^{kj}$, for $j\in \textbf{S}_{k}$. The user obtains these subfiles in two stages. In the first stage, the user obtains subfiles $W_{d_{k}}^{jk}$, for $j\in \textbf{S}_{k}$. One of the packet available in $X_{\textbf{\textit{d}}}$ is,
	\begin{equation}
	\label{RC_stage1_RX}
	X_{\textbf{\textit{d}}}^{j}=\bigg(\frac{\alpha_{j}^{k}}{N_{j}^{k}}\bigg)W^{jk}_{d_{k}}+\sum_{s\in \textbf{S}_{j}\setminus\{k\}}\bigg(\frac{\alpha_{j}^{s}}{N_{j}^{s}}\bigg)W^{js}_{d_{s}}.
	\end{equation}
	Since subfiles $W^{js}_{n}$, for $s\in \textbf{S}_{j}\setminus\{k\}=\textbf{S}_{k}\setminus\{j\}$, are available in $Z_{k}$, the user can evaluate
	\begin{equation}
	\label{RC_stage1_ob}
	\sum_{s\in \textbf{S}_{j}\setminus\{k\}}\bigg(\frac{\alpha_{j}^{s}}{N_{j}^{s}}\bigg)W^{js}_{d_{s}}
	\end{equation}
	 The subfile $W_{d_{k}}^{jk}$ can computed from (\ref{RC_stage1_RX}) and (\ref{RC_stage1_ob}). In the second stage, the user recovers subfiles $W_{d_{k}}^{kj}$, for $j\in \textbf{S}_{k}$.  Another packet available in $X_{\textbf{\textit{d}}}$ is,
	\begin{equation}
	\label{RC_stage2_rc}
	X_{\textbf{\textit{d}}}^{k}=\sum_{j\in \textbf{S}_{k}}\bigg(\frac{\alpha_{k}^{j}}{N_{k}^{j}}\bigg) W^{kj}_{d_{j}}
	\end{equation}
	Since $W_{n}^{k(k+1)}-W_{n}^{kj}$, for $j\in \textbf{S}_{k}$, are available in $Z_{k}$, the user can evaluate
	\begin{equation}
	\label{recovery1}
	\sum_{j\in \textbf{S}_{k}}\bigg(\frac{\alpha_{k}^{j}}{N_{k}^{j}}\bigg)(W_{d_{j}}^{k(k+1)}-W_{d_{j}}^{kj})
	\end{equation} 
	Combining (\ref{RC_stage2_rc}) and (\ref{recovery1}) the user can compute
	\begin{equation}
	\sum_{j\in \textbf{S}_{k}}\bigg(\frac{\alpha_{k}^{j}}{N_{k}^{j}}\bigg)W_{d_{j}}^{k(k+1)}
	\end{equation}
	This  can be rewritten as
	\begin{equation}
	\bigg(\sum_{n\in [N]\setminus{d_{k}}}\sum_{j\in \textbf{S}_{k}: d_{j}=n}\bigg(\frac{\alpha_{k}^{j}}{N_{k}^{j}}\bigg)W_{d_{j}}^{k(k+1)}\bigg)+\bigg(\sum_{j\in \textbf{S}_{k}:d_{j}=d_{k}}\bigg(\frac{\alpha_{k}^{j}}{N_{k}^{j}}\bigg) W_{d_{j}}^{k(k+1)}\bigg)
	\end{equation}
	Note that when $W_{d_{j}}\ne W_{d_{k}}$,  $\alpha_{k}^{j}=1$ and when $W_{d_{j}}=W_{d_{k}}$, $\alpha_{k}^{j}=-1$. Recall that $N_{k}^{k}$ denotes the number of users request for file $W_{d_{k}}$ in set $\textbf{S}_{k}$. Now the above expression simplifies to
	\begin{equation}
	\label{recover2}
	\bigg(\sum_{n\in [N]\setminus{d_{k}}}W_{n}^{k(k+1)}\bigg)-I\{N_{k}^{k}\ne 0\} W_{d_{k}}^{k(k+1)}
	\end{equation}
	With the help of the cached function
	\begin{equation}
	\sum_{n\in [N]}W_{n}^{k(k+1)}=\bigg(\sum_{n\in [N]\setminus\{d_{k}\}}W_{n}^{k(k+1)}\bigg)+W_{d_{k}}^{k(k+1)}
	\end{equation} 
	 and (\ref{recover2}) the user can compute the subfile $W_{d_{k}}^{k(k+1)}$. Combining this with $W_{d_{k}}^{k(k+1)}-W_{d_{k}}^{kj}$ available in $Z_{k}$  user can obtain subfiles $W_{d_{k}}^{kj}$. Using all the recovered subfiles the user can reconstruct the requested file $W_{d_{k}}$.

	 The above observations can be summarised as:
	\begin{theorem}
		\label{new rate}
		The memory rate pair $\Big(\frac{N}{K}\big(K-2+\frac{(K-2+1/N)}{(K-1)}\big),\frac{1}{K-1}\Big)$ is achievable by symmetric caching schemes for the $(N,K)$ cache network.
	\end{theorem}
	\noindent The caching scheme proposed in \cite{maddah2014fundamental,yu2016exact} achieves the memory rate pair $\big(\frac{N}{K}(K-1),\frac{1}{K}\big)$, and by memory sharing between that scheme and the proposed scheme, we can  achieve all memory rate pairs $\big(M,\frac{(KN-1)}{K(N-1)}-\frac{1}{(N-1)}M\big)$, where $M\in \big[\frac{N}{K}\big(K-2+\frac{(K-2+1/N)}{(K-1)}\big),\frac{N}{K}(K-1)\big]$. By deriving a matching lower bound, we show that this is the exact rate memory tradeoff when $\big\lceil\frac{K+1}{2}\big\rceil\leq N\leq K$.

	\subsection{Matching Lower Bound}

	\noindent Consider the demand
	\begin{equation}
	\textbf{\textit{d}}_{1}=(W_{1},W_{2},\dots, W_{N},W_{1},W_{2},\dots, W_{K-N}).
	\end{equation}
	Demands $\{\textbf{\textit{d}}_{l}:2\leq l\leq K\}$, are obtained from the demand $\textbf{\textit{d}}_{1}$ by cyclic left shifts as shown in TABLE \ref{demands}. Consider the demand $\textbf{\textit{b}}_{l}$ defined as
	\begin{equation*}
	\textbf{\textit{b}}_{l}=\left\{\begin{aligned}
	\textbf{\textit{d}}_{N-l+1}, &\text{ for $1\leq l\leq N$}\\
	\textbf{\textit{d}}_{K+N-l+1}, &\text{ for $N+1\leq l\leq K$}
	\end{aligned}\right.
	\end{equation*}
	\begin{table*}[tp]
		\centering
		\setlength{\tabcolsep}{3pt}
		\begin{tabular}{r|c|c|c|>{\columncolor{gray!30!white}}c|c|c|c|c|>{\columncolor{gray!30!white}}c|c|c|}
			\hhline{~-----------} 
			&Demands &$U_{1}$&\dots&$U_{i}$&\dots&$U_{N}$&$U_{N+1}$&\dots&$U_{N+i}$&\dots&$U_{K}$\\[0.2em]
			\hhline{~-----------} 
			\ldelim\{{8}{11.5mm}[\parbox{4.5mm}{ {${\textbf{\textit{A}}_{i}}$}}]&$ {\textbf{\textit{d}}_{1}}$&$ {W_{1}}$& {\dots}&$ {W_{i}}$& {\dots}&$ {W_{N}}$&$ {W_{1}}$& {\dots}& {$W_{i}$}& {\dots}&$ {W_{K-N}}$\\[0.2em] 
			\hhline{~-----------} 
			&$ {\textbf{\textit{d}}_{2}}$&$ {W_{2}}$& {\dots}& {$W_{i+1}$}& {\dots}&$ {W_{1}}$&$ {W_{2}}$& {\dots}& {$W_{i+1}$}& {\dots}&$ {W_{1}}$\\[0.2em] 
			\hhline{~-----------} 
			& {\dots}& {\dots}& {\dots}& {\dots}& {\dots}& {\dots}& {\dots}& {\dots}& {\dots}& {\dots}& {\dots}\\[0.2em]
			
			\hhline{~-----------}  &$ {\textbf{\textit{d}}_{\substack{K-N-i+1}}}$&$ {W_{\substack{K-N-i+1}}}$& {\dots}& {$W_{K-N}$}& {\dots}&$ {W_{\substack{K-N-i}}}$&$ {W_{\substack{K-N-i+1}}}$& {\dots}& {$W_{K-N}$}& {\dots}&$ {W_{\substack{K-N-i}}}$\\[0.2em] 
			\hhline{~-----------} 
			\ldelim\{{3.5}{5mm}[\parbox{3.5mm}{ {${\textbf{\textit{C}}_{i}}$}}]&$ {\textbf{\textit{d}}_{\substack{K-N-i+2}}}$&$ {W_{\substack{K-N-i+2}}}$& {\dots}& {$W_{K-N+1}$}& {\dots}&$ {W_{\substack{K-N-i+1}}}$&$ {W_{\substack{K-N-i+2}}}$& {\dots}& {$W_{1}$}& {\dots}&$ {W_{\substack{K-N-i-1}}}$\\[0.2em] 
			\hhline{~-----------} 
			
			& {\dots}& {\dots}& {\dots}& {\dots}& {\dots}& {\dots}& {\dots}& {\dots}& {\dots}& {\dots}& {\dots}\\[0.2em]
			\hhline{~-----------}  &$ {\textbf{\textit{d}}_{\substack{N-i}}}$&$ {W_{\substack{N-i}}}$& {\dots}& {$W_{N-1}$}& {\dots}&$ {W_{\substack{2N-K-i-1}}}$&$ {W_{\substack{2N-K-i}}}$& {\dots}& {$W_{2N-K-1}$}& {\dots}&$ {W_{\substack{N-i-1}}}$\\[0.2em]
			\hhline{~-----------} &$\textbf{\textit{d}}_{\substack{N-i+1}}$&$W_{\substack{N-i+1}}$&\dots&\cellcolor{LimeGreen!50!white}{$W_{N}$}&\dots&$W_{\substack{2N-K-i}}$&$W_{\substack{2N-K-i+1}}$&\dots&$W_{2N-K}$&\dots&$W_{\substack{N-i}}$\\[0.2em] \hhline{~-----------} 
			&\dots&\dots&\dots&\dots&\dots&\dots&\dots&\dots&\dots&\dots&\dots\\[0.2em]
			
			\hhline{~-----------}  &$\textbf{\textit{d}}_{\substack{N}}$&$W_{\substack{N}}$&\dots&$W_{i-1}$&\dots&$W_{\substack{2N-K}}$&$W_{\substack{2N-K+1}}$&\dots&$W_{\substack{2N-K+i}}$&\dots&$W_{\substack{N-1}}$\\[0.2em] \hhline{~-----------} 
			\ldelim\{{7}{10.5mm}[\parbox{3.5mm}{ {${\textbf{\textit{J}}}$}}]&$ {\textbf{\textit{d}}_{\substack{N+1}}}$&$ {W_{\substack{1}}}$& {\dots}& {$W_{i}$}& {\dots}&$ {W_{\substack{2N-K+1}}}$&$ {W_{\substack{2N-K+2}}}$& {\dots}& {$W_{\substack{2N-K+i+1}}$}& {\dots}&$ {W_{\substack{N}}}$\\[0.2em] \hhline{~-----------} 
			& {\dots}& {\dots}& {\dots}& {\dots}& {\dots}& {\dots}& {\dots}& {\dots}& {\dots}& {\dots}& {\dots}\\[0.2em]
			
			\hhline{~-----------} &$ {\textbf{\textit{d}}_{\substack{K-i+1}}}$&$ {W_{\substack{K-N-i+1}}}$& {\dots}& {$W_{K-N}$}& {\dots}&$ {W_{\substack{N-i}}}$&$ {W_{\substack{N-i+1}}}$& {\dots}&\cellcolor{LimeGreen!50!white} {$W_{\substack{N}}$}& {\dots}&$ {W_{\substack{K-N-i}}}$\\[0.2em] 
			\hhline{~-----------} 
			\ldelim\{{3.5}{5mm}[\parbox{3.5mm}{ {${\textbf{\textit{B}}_{i}}$}}]&$ {\textbf{\textit{d}}_{\substack{K-i+2}}}$&$ {W_{\substack{K-N-i+2}}}$& {\dots}& {$W_{1}$}& {\dots}&$ {W_{\substack{N-i+1}}}$&$ {W_{\substack{N-i+2}}}$& {\dots}& {$W_{\substack{1}}$}& {\dots}&$ {W_{\substack{K-N-i+1}}}$\\[0.2em] 
			\hhline{~-----------} 
			& {\dots}& {\dots}& {\dots}& {\dots}& {\dots}& {\dots}& {\dots}& {\dots}& {\dots}& {\dots}& {\dots}\\[0.2em]
			
			\hhline{~-----------}  &$ {\textbf{\textit{d}}_{K}}$&$ {W_{\substack{K-N}}}$& {\dots}& {$W_{i-1}$}& {\dots}&$ {W_{\substack{N-1}}}$&$ {W_{\substack{N}}}$& {\dots}& {$W_{\substack{i-1}}$}& {\dots}&$ {W_{\substack{K-N-1}}}$\\[0.2em] 
			\hhline{~-----------}
			
		\end{tabular}
		\caption{The set of demands $\{\textbf{\textit{d}}_{l}: 1\leq l\leq K\}$}
		\label{demands}
	\end{table*}
	\noindent It can be noted that in demand $\textbf{\textit{b}}_{l}$, the user $U_l$ requests for the file $W_{N}$. Let $X_{\textbf{\textit{d}}_{l}}$ denote the set of packets broadcast by the server in response to the demand $\textbf{\textit{d}}_{l}$. For $\textbf{\textit{S}}\subseteq \{\textbf{\textit{d}}_{1},\dots,\textbf{\textit{d}}_{K}\}$, let $X_{\textbf{\textit{S}}}$ denote the set of all packets broadcast in response to the demands in the set ${\textbf{\textit{S}}}$. The following lemma are easy to obtain:

	\begin{lemma}
		\label{re}
		For $\textbf{S}$, $\textbf{\textit{T}}\subset\{\textbf{\textit{d}}_{1},\dots, \textbf{\textit{d}}_{K}\}\setminus\{\textbf{\textit{b}}_{l}\}$, we have the identity,
		\begin{IEEEeqnarray*}{rl}
		H(W_{[N-1]},X_{\textbf{S}\cup\{\textbf{\textit{b}}_{l}\}})+H(W_{[N-1]},Z_{l}&,X_{\textbf{\textit{T}}})\geq H(W_{[N-1]},X_{\textbf{S}\cap \textbf{\textit{T}}})+N
		\end{IEEEeqnarray*}
	\end{lemma}
	\begin{proof}
		We have,
		\begin{IEEEeqnarray*}{rl}
		H(W_{[N-1]},X_{\textbf{\textit{S}}\cup\{\textbf{\textit{b}}_{l}\}})+H(W_{[N-1]},Z_{l},X_{\textbf{\textit{T}}})&\overset{(a)}{\geq} H(W_{[N-1]},X_{\textbf{\textit{S}}\cap \textbf{\textit{T}}})+H(W_{[N-1]},Z_{l},X_{\textbf{\textit{S}}\cup \textbf{\textit{T}}\cup\{\textbf{\textit{b}}_{l}\}})\\
		&\overset{(b)}{=}H(W_{[N-1]},X_{\textbf{\textit{S}}\cap \textbf{\textit{T}}})+H(W_{[N-1]},W_{N},Z_{l},X_{\textbf{\textit{S}}\cup \textbf{\textit{T}}\cup\{\textbf{\textit{b}}_{l}\}})\\
		&\overset{(c)}{=}H(W_{[N-1]},X_{\textbf{\textit{S}}\cap \textbf{\textit{T}}})+H(W_{[N]})
		\overset{}{=}H(W_{[N-1]},X_{\textbf{\textit{S}}\cap \textbf{\textit{T}}})+N
		\end{IEEEeqnarray*}
		where
		
		\hspace{-4mm}\begin{tabular}{cl}
			$(a)$& follows from the submodularity property of entropy,\\ $(b)$& follows from (\ref{I.1}), \\$(c)$& follow from (\ref{I.2}).
		\end{tabular}
		
	\end{proof}
	
	\begin{lemma}
		\label{sum increment}
		For a sequence of sets  $\textbf{S}_{i}\subset\{\textbf{\textit{d}}_{1},\dots, \textbf{\textit{d}}_{K}\}\setminus\{\textbf{\textit{b}}_{i}\}$, such that $\textbf{S}_{i}=\textbf{S}_{i+1}\cup \{\textbf{\textit{b}}_{i+1}\}$, we have the identity
		\begin{IEEEeqnarray*}{rl}
		H(W_{[N-1]},X_{\textbf{S}_{l}})+\sum_{i=l+1}^{j}H(W_{[N-1]},Z_{i},X_{\textbf{S}_{i}})\geq H(W_{[N-1]},X_{\textbf{S}_{j}})+(j-l)N
		\end{IEEEeqnarray*}
	\end{lemma}
	\begin{proof}
		We have,
		\begin{IEEEeqnarray*}{rl}
		H(W_{[N-1]}&,X_{\textbf{\textit{S}}_{l}})+\sum_{i=l+1}^{j}H(W_{[N-1]},Z_{i},X_{\textbf{\textit{S}}_{i}})=H(W_{[N-1]},X_{\textbf{\textit{S}}_{l}})+H(W_{[N-1]},Z_{l+1},X_{\textbf{\textit{S}}_{l+1}})\\&\text{ }+\sum_{i=l+2}^{j}H(W_{[N-1]},Z_{i},X_{\textbf{\textit{S}}_{i}})\\
		&=\Big(H(W_{[N-1]},X_{\textbf{\textit{S}}_{l+1}},X_{\textbf{\textit{b}}_{{l+1}}})+H(W_{[N-1]},Z_{l+1},X_{\textbf{\textit{S}}_{l+1}})\Big)+\sum_{i=l+2}^{j}H(W_{[N-1]},Z_{i},X_{\textbf{\textit{S}}_{i}})\\
		&\overset{(a)}{\geq}H(W_{[N]})+\Big(H(W_{[N-1]},X_{\textbf{\textit{S}}_{l+1}})+H(W_{[N-1]},Z_{l+2},X_{\textbf{\textit{S}}_{l+2}})\Big)+\sum_{i=l+3}^{j}H(W_{[N-1]},Z_{i},X_{\textbf{\textit{S}}_{i}})\\
		&\overset{(a)}{\geq}2H(W_{[N]})+H(W_{[N-1]},X_{\textbf{\textit{S}}_{l+2}})+H(W_{[N-1]},Z_{l+3},X_{\textbf{\textit{S}}_{l+3}})+\sum_{i=l+4}^{j}H(W_{[N-1]},Z_{i},X_{\textbf{\textit{S}}_{i}})\\
		&\overset{(b)}{\geq}(j-l)H(W_{[N]})+H(W_{[N-1]},X_{\textbf{\textit{S}}_{j}})
		\overset{}{=}H(W_{[N-1]},X_{\textbf{\textit{S}}_{j}})+(j-l)N
		\end{IEEEeqnarray*}
		where
		
		\hspace{-4mm}\begin{tabular}{cl}
			$(a)$& follows from Lemma \ref{re} with  $\textbf{\textit{S}}=\textbf{\textit{T}}=\textbf{\textit{S}}_{l+1}$,\\
			$(b)$& follows from repeated use of Lemma \ref{re} with $\textbf{\textit{S}}\cup\{\textbf{\textit{b}}_{l}\}=\textbf{\textit{S}}_{i}$ and $\textbf{\textit{T}}=\textbf{\textit{S}}_{i+1}$ for $l+3\leq i\leq j$.
		\end{tabular}
	
	\end{proof}
	\noindent In a similar fashion, for a sequence of sets  $\textbf{\textit{T}}_{i}\subset\{\textbf{\textit{d}}_{1},\dots, \textbf{\textit{d}}_{K}\}\setminus\{\textbf{\textit{b}}_{i}\}$, such that $\textbf{\textit{T}}_{i}=\textbf{\textit{T}}_{i-1}\cup \{\textbf{\textit{b}}_{i-1}\}$, we can obtain
	\begin{equation}
	\label{sum decrement}
	H(W_{[N-1]},X_{\textbf{\textit{T}}_{j+1}})+\sum_{i=l}^{j}H(W_{[N-1]},Z_{i},X_{\textbf{\textit{T}}_{i}})\geq H(W_{[N-1]},X_{\textbf{\textit{T}}_{l}})+(j-l+1)N
	\end{equation}
	For $1\leq i\leq N$, let us consider the sets of demands as shown below:
	\begin{table}[h!]
		\centering
		\begin{tabular}{|c|c|c|c|}
			\hline
			Set &Demands &Number& {Files Requested by  $U_i$}\\
			\hline
			 {${\textbf{\textit{A}}_{i}}$}& {${\textbf{\textit{d}}}_{1},\dots, \textbf{\textit{d}}_{N-i}$}& $N-i$& {$W_{i},\dots, W_{(N-1)}$}\\
			\hline
			 {${\textbf{\textit{B}}_{i}}$}& {$\textbf{\textit{d}}_{K-i+2},\dots, \textbf{\textit{d}}_{K}$}&$i-1$& {$W_{1},\dots, W_{i-1}$}\\
			\hline
			 {${\textbf{\textit{C}}_{i}}$}& {$\textbf{\textit{d}}_{K-N-i+2},\dots, \textbf{\textit{d}}_{N-i}$}&$2N-K-1$& {$W_{(K-N+1)},\dots,W_{N-1}$}\\
			\hline
			 {${\textbf{\textit{J}}}$}& {$\textbf{\textit{d}}_{N+1},\dots, \textbf{\textit{d}}_{K}$}&$K-N$& {$W_{1},\dots, W_{K-N}$}\\
			\hline
		\end{tabular}
	\end{table}
	
	 \noindent These set are also indicated in TABLE \ref{demands}. Note that
	\begin{IEEEeqnarray}{rl}
	\label{case 1 null set}
	{\textbf{\textit{A}}_{N}}={\textbf{\textit{B}}_{1}}={\textbf{\textit{C}}_{N}}&=\phi\\
	\label{case1 Set A}
	{\textbf{\textit{A}}_{i+1}}\cup \{\textbf{\textit{b}}_{i+1}\}& ={\textbf{\textit{A}}_{i}}\\
	\label{case1 Set B}
	{\textbf{\textit{B}}_{i}}\cup\{\textbf{\textit{b}}_{N+i}\}&={\textbf{\textit{B}}_{i+1}}\\
	\label{case 1 set AC}
	{\textbf{\textit{A}}_{i}}\cap {\textbf{\textit{C}}_{i}}&={\textbf{\textit{C}}_{i}}\\
	\label{case 1 DB K-N}
	\textbf{\textit{B}}_{K-N}\cup \{\textbf{\textit{b}}_{K}\}=\textbf{\textit{B}}_{K-N+1}&=\textbf{\textit{J}}
	\end{IEEEeqnarray}
	\begin{equation}
	\label{case1 Set BD}
	{\textbf{\textit{B}}_{i}}\cap{\textbf{\textit{J}}}=\left\{ \text{ }\begin{aligned}
	{\textbf{\textit{B}}_{i}} &\text{ when $1\leq i\leq K-N$}\\
	{\textbf{\textit{J}}}&\text{ when $K-N+1\leq i\leq N$}
	\end{aligned}\right.
	\end{equation}

	\noindent It can be noted that in the demand set ${\textbf{\textit{B}}_{i}}$, both  users $U_i$ and $U_{N+i}$ are requesting for the same set of files (for $1\leq i\leq K-N$). Note that $\mid \textbf{\textit{A}}_{i}\cup \textbf{\textit{B}}_{i}\mid = \mid \textbf{\textit{J}}\cup \textbf{\textit{C}}_{i}\mid=N-1$. Thus, we have
	\begin{align}
	\label{case1:M,A,B}
	M+(N-1)R\geq H(Z_{i})+H(X_{\textbf{\textit{A}}_{i}\cup \textbf{\textit{B}}_{i}})\geq H(Z_{i},X_{\textbf{\textit{A}}_{i}\cup \textbf{\textit{B}}_{i}})
	\end{align}
	Similarly,
	\begin{align}
	\label{case1:J,C}
	M+(N-1)R\geq H(Z_{i})+H(X_{\textbf{\textit{J}}\cup \textbf{\textit{C}}_{i}})\geq H(Z_{i},X_{\textbf{\textit{J}}\cup \textbf{\textit{C}}_{i}})
	\end{align}
	Now we can obtain the following result:
	\begin{theorem}
		\label{Main_Result_1}
		For the $(N, K)$ cache network, when $\big\lceil\frac{K+1}{2}\big\rceil\leq N\leq K$, achievable memory rate pairs $(M,R)$ must satisfy the constraint
		\begin{equation*}
		KM+K(N-1)R\geq KN-1.
		\end{equation*}
	\end{theorem}
	
	\begin{proof}
		We have,
		\vspace{-3mm}
		\begin{IEEEeqnarray*}{rl}
		KM+&K(N-1)R=N(M+(N-1)R)+(K-N)(M+(N-1)R)\\
		\overset{(a)}{\geq}& \sum_{i=1}^{N}H(Z_{i},X_{\textbf{\textit{A}}_{i}\cup\textbf{\textit{B}}_{i}})+\sum_{i=1}^{K-N}H(Z_{i},X_{\textbf{\textit{J}}\cup\textbf{\textit{C}}_{i}})\\
		\overset{(b)}{=}& \sum_{i=1}^{N}H(W_{[N-1]},Z_{i},X_{\textbf{\textit{A}}_{i}\cup\textbf{\textit{B}}_{i}})+\sum_{i=1}^{K-N}H(W_{[N-1]},Z_{i},X_{\textbf{\textit{J}}\cup\textbf{\textit{C}}_{i}})\\
		\overset{}{=}& \sum_{i=1}^{K-N}\Big(H(W_{[N-1]},Z_{i},X_{\textbf{\textit{A}}_{i}\cup\textbf{\textit{B}}_{i}})+H(W_{[N-1]},Z_{i},X_{\textbf{\textit{J}}\cup\textbf{\textit{C}}_{i}})\Big)+\sum_{j=K-N+1}^{N}H(W_{[N-1]},Z_{j},X_{\textbf{\textit{A}}_{j}\cup\textbf{\textit{B}}_{j}})\\
		\overset{(c)}{\geq}& \sum_{i=1}^{K-N}\Big(H(W_{[N-1]},Z_{i},X_{\textbf{\textit{A}}_{i}\cup\textbf{\textit{J}}})+H(W_{[N-1]},Z_{i},X_{\textbf{\textit{B}}_{i}\cup\textbf{\textit{C}}_{i}})\Big)+\sum_{j=K-N+1}^{N}H(W_{[N-1]},Z_{j},X_{\textbf{\textit{B}}_{j}\cup\textbf{\textit{A}}_{j}})\\
		\overset{}{\geq}& \left(H(W_{[N-1]},X_{\textbf{\textit{A}}_{1}\cup\textbf{\textit{J}}})+\sum_{i=2}^{K-N}H(W_{[N-1]},Z_{i},X_{\textbf{\textit{A}}_{i}\cup\textbf{\textit{J}}})\right)+\sum_{j=K-N+1}^{N}H(W_{[N-1]},Z_{j},X_{\textbf{\textit{B}}_{j}\cup\textbf{\textit{A}}_{j}})\\&+\sum_{i=1}^{K-N}H(W_{[N-1]},Z_{i},X_{\textbf{\textit{B}}_{i}\cup\textbf{\textit{C}}_{i}})\\
		\overset{(d)}{\geq}& (K-N-1)N+\left(H(W_{[N-1]},X_{\textbf{\textit{A}}_{K-N}\cup\textbf{\textit{J}}})+\sum_{j=K-N+1}^{N}H(W_{[N-1]},Z_{j},X_{\textbf{\textit{B}}_{j}\cup\textbf{\textit{A}}_{j}})\right)\\&+\sum_{i=1}^{K-N}H(W_{[N-1]},Z_{i},X_{\textbf{\textit{B}}_{i}\cup\textbf{\textit{C}}_{i}})\\
		\overset{(e)}{\geq}& (N-1)N+H(W_{[N-1]},X_{\textbf{\textit{A}}_{N}\cup\textbf{\textit{J}}})+\sum_{i=1}^{K-N}H(W_{[N-1]},Z_{i},X_{\textbf{\textit{B}}_{i}\cup\textbf{\textit{C}}_{i}})\\
		\overset{}{\geq}& (N-1)N+H(W_{[N-1]},X_{\textbf{\textit{A}}_{N}\cup\textbf{\textit{J}}})+\sum_{i=1}^{K-N}H(W_{[N-1]},Z_{i},X_{\textbf{\textit{B}}_{i}})\\
		\overset{(f)}{=}& (N-1)N+H(W_{[N-1]},X_{\textbf{\textit{J}}})+\sum_{i=1}^{K-N}H(W_{[N-1]},Z_{i},X_{\textbf{\textit{B}}_{i}})\\
		\overset{(g)}{=}& (N-1)N+H(W_{[N-1]},X_{\textbf{\textit{J}}}) +\sum_{i=1}^{K-N}H(W_{[N-1]},Z_{i+N},X_{\textbf{\textit{B}}_{i}})\\
		\overset{(h)}{=}& (N-1)N+\left(H(W_{[N-1]},X_{\textbf{\textit{B}}_{K-N+1}}) +\sum_{i=1}^{K-N}H(W_{[N-1]},Z_{i+N},X_{\textbf{\textit{B}}_{i}})\right)\\
		\overset{(i)}{\geq}& (N-1)N+(K-N)N+H(W_{[N-1]},X_{\textbf{\textit{B}}_{1}})\\
		\overset{(f)}{=}& (K-1)N+H(W_{[N-1]})
		\overset{}{\geq} KN-1
		\end{IEEEeqnarray*}
		\noindent where
		
		\hspace{-6mm}	\begin{tabular}{rl}
				$(a)$& follows from (\ref{case1:M,A,B}) and (\ref{case1:J,C}),
			\end{tabular}\\
		\hspace{-6mm}	\begin{tabular}{rl}
				$(b)$& follows from (\ref{I.1}) and the definition of sets $\textbf{\textit{A}}_{i}$, $\textbf{\textit{B}}_{i}$, $\textbf{\textit{C}}_{i}$ and $\textbf{\textit{J}}$,\\
				$(c)$& follows from the facts that $\textbf{\textit{A}}_{i}\cap \textbf{\textit{C}}_{i}=\textbf{\textit{C}}_{i}$, $\textbf{\textit{J}}\cap \textbf{\textit{B}}_{i}=\textbf{\textit{B}}_{i}$ for $1\leq i\leq K-N$ \\&(refer (\ref{case 1 set AC}) and (\ref{case1 Set BD})) and the submodularity property of entropy,\\
				$(d)$& follows from Lemma \ref{sum increment} with $\textbf{\textit{S}}_{i}=\textbf{\textit{A}}_{i}\cup \textbf{\textit{J}}$, $l=1$, $j=K-N$ and (\ref{case1 Set A}),\\
				$(e)$& follows from Lemma \ref{sum increment}, with $\textbf{\textit{S}}_{i}=\textbf{\textit{A}}_{i}\cup \textbf{\textit{J}}$, $l=K-N$, $j=N$,  the fact that $\textbf{\textit{J}}\cap \textbf{\textit{B}}_{i}=\textbf{\textit{J}}$\\& for $K-N+1\leq i\leq K$(refer (\ref{case1 Set BD})) and (\ref{case1 Set A}),\\
				$(f)$& follows from (\ref{case 1 null set}),\\
				$(g)$& follows from  (\ref{symmetry}),\\
				$(h)$& follows from (\ref{case 1 DB K-N}),\\
				$(i)$& follows from (\ref{sum decrement}) with $\textbf{\textit{T}}_{i}=\textbf{\textit{B}}_{i}$, $l=1$, $j=K-N$ and (\ref{case1 Set B}).
			\end{tabular}
	
\end{proof}
	The above observations improve upon the previous results from \cite{maddah2014fundamental,yu2016exact} as shown in TABLE \ref{Table:Case1}. 
	\begin{table}[ht]
		\centering
		\setlength{\tabcolsep}{1.8pt}
		\begin{tabular}{|c|c|c|c|c|}
			\hline
			Memory &Rate \cite{maddah2014fundamental,yu2016exact}&Lower Bound \cite{maddah2014fundamental,yu2016exact}&New Rate& New Lower Bound\\
			\hline&&&&\\[-1.2em]
			$\frac{N}{K}\big(K-2+\frac{(K-2+1/N)}{(K-1)}\big)\leq M \leq \frac{N(K-1)}{K}$&$
			\frac{(K^{2}+K-2)}{K(K-1)}-\frac{(K+1)M}{N(K-1)}$&$R\geq1-\frac{M}{N}$&$\frac{(KN-1)}{K(N-1)}-\frac{M}{(N-1)}$&$R\geq \frac{(KN-1)}{K(N-1)}-\frac{M}{(N-1)}$\\[0.4em]
			\hline
		\end{tabular}
		\caption{Rate memory tradeoff when $\big\lceil \frac{K+1}{2} \big\rceil\leq N \leq K$}
		\label{Table:Case1}
	\end{table}

	\noindent We summarise as:
	\begin{theorem}
		For the $(N,K)$ cache network, when $\big\lceil\frac{K+1}{2}\big\rceil\leq N\leq K$,  
		the exact rate memory tradeoff is given by
		\begin{equation}
		R^{*}(M)=\frac{(KN-1)}{K(N-1)}-\frac{1}{(N-1)}M
		\end{equation}
		where $M\geq \frac{N}{K}\big(K-2+\frac{(K-2+1/N)}{(K-1)}\big)$.
	\end{theorem}

	\section[Case II:]{Case II: $1\leq N\leq \lceil\frac{K+1}{2}\rceil$}
	
	  The caching scheme proposed by Yu et al. in \cite{yu2016exact} can achieve all memory rate pairs 
	
	  \begin{equation}
	  \left(M,\frac{K^{2}+K-2}{K(K-1)}-\frac{(K+1)}{N(K-1)}M\right),
	  \end{equation}
	  where $M\in \left[\frac{N(K-2)}{K},\frac{N(K-1)}{K}\right]$, for the $(N,K)$ cache network. By deriving a matching lower bound, we show that this is the exact rate memory tradeoff when $1\leq N\leq \big\lceil\frac{K+1}{2}\big\rceil$. 
	 Consider the demand
	 \begin{equation}
	 \textbf{\textit{d}}_{1}=(W_{1},W_{2},\dots, W_{N},W_{1},W_{2},\dots, W_{N-1}, W_{1},W_{1},\dots,W_{1})
	 \end{equation}
	 Demands $\{\textbf{\textit{d}}_{l}:2\leq l\leq K\}$, are obtained from the demand $\textbf{\textit{d}}_{1}$ by  cyclic left shifts as shown in TABLE \ref{demands case 2}.
	\begin{table*}[tp]
		\centering
		\setlength{\tabcolsep}{2.3pt}
		\begin{tabular}{r|c|c|c|>{\columncolor{gray!30!white}}c|c|c|c|c|>{\columncolor{gray!30!white}}c|c|c|c|c|>{\columncolor{gray!30!white}}c|c|c|}
			\hhline{~----------------}
			&Demand&$U_{1}$&\dots&$U_{i}$&\dots&$U_{N}$&$U_{N+1}$&\dots&$U_{N+i}$&\dots&$U_{2N-1}$&$U_{2N}$&\dots&$U_{j}$&\dots&$U_{K}$\\
			\hhline{~----------------}
			\ldelim\{{4}{5mm}[\parbox{3.5mm}{${\textbf{\textit{A}}_{i}}$}]&$\textbf{\textit{d}}_{1}$&$W_{1}$&\dots&$W_{i}$&\dots&$W_{N}$&$W_{1}$&\dots&$W_{i}$&\dots&$W_{N-1}$&$W_{1}$&\dots&$W_{1}$&\dots&$W_{1}$\\
			\hhline{~----------------}
			&$\textbf{\textit{d}}_{2}$&$W_{2}$&\dots&$W_{i+1}$&\dots&$W_{1}$&$W_{2}$&\dots&$W_{i+1}$&\dots&$W_{1}$&$W_{1}$&\dots&$W_{1}$&\dots&$W_{1}$\\
			\hhline{~----------------}
			&\dots&\dots&\dots&\dots&\dots&\dots&\dots&\dots&\dots&\dots&\dots&\dots&\dots&\dots&\dots&\dots\\
			\hhline{~----------------}
			&$\textbf{\textit{d}}_{N-i}$&$W_{N-i}$&\dots&$W_{N-1}$&\dots&$W_{\substack{N-i-1}}$&$W_{N-i}$&\dots&$W_{N-1}$&\dots&$W_{1}$&$W_{1}$&\dots&$W_{1}$&\dots&$W_{i-1}$\\
			\hhline{~----------------}
			&$\textbf{\textit{d}}_{\substack{N-i+1}}$&$W_{\substack{N-i+1}}$&\dots&\cellcolor{LimeGreen!50!white}${W_{N}}$&\dots&$W_{N-i}$&$W_{\substack{N-i+1}}$&\dots&$W_{1}$&\dots&$W_{1}$&$W_{1}$&\dots&$W_{1}$&\dots&$W_{i}$\\
			\hhline{~----------------}
			&\dots&\dots&\dots&\dots&\dots&\dots&\dots&\dots&\dots&\dots&\dots&\dots&\dots&\dots&\dots&\dots\\
			\hhline{~----------------}
			\ldelim\{{3}{5mm}[\parbox{3.5mm}{${\textbf{\textit{E}}_{i}}$}]&$\textbf{\textit{d}}_{N+1}$&$W_{1}$&\dots&$W_{i}$&\dots&$W_{1}$&$W_{1}$&\dots&$W_{1}$&\dots&$W_{1}$&$W_{1}$&\dots&$W_{1}$&\dots&$W_{N}$\\
			\hhline{~----------------}
			&\dots&\dots&\dots&\dots&\dots&\dots&\dots&\dots&\dots&\dots&\dots&\dots&\dots&\dots&\dots&\dots\\
			\hhline{~----------------}
			&$\textbf{\textit{d}}_{2N-i}$&$W_{N-i}$&\dots&$W_{N-1}$&\dots&$W_{1}$&$W_{1}$&\dots&$W_{1}$&\dots&$W_{1}$&$W_{1}$&\dots&$W_{1}$&\dots&$W_{\substack{N-i-1}}$\\
			\hhline{~----------------}
			\ldelim\{{11.3}{5mm}[\parbox{3.5mm}{${\textbf{\textit{G}}_{i}}$}]&$\textbf{\textit{d}}_{\substack{2N-i+1}}$&$W_{\substack{N-i+1}}$&\dots&$W_{1}$&\dots&$W_{1}$&$W_{1}$&\dots&$W_{1}$&\dots&$W_{1}$&$W_{1}$&\dots&$W_{1}$&\dots&$W_{N-i}$\\
			\hhline{~----------------}
			&\dots&\dots&\dots&\dots&\dots&\dots&\dots&\dots&\dots&\dots&\dots&\dots&\dots&\dots&\dots&\dots\\
			
			\hhline{~----------------}
			&$\textbf{\textit{d}}_{\substack{K-j+2}}$&$W_{1}$&\dots&$W_{1}$&\dots&$W_{1}$&$W_{1}$&\dots&$W_{1}$&\dots&$W_{1}$&$W_{1}$&\dots&$W_{1}$&\dots&$W_{1}$\\
			\hhline{~----------------}
			&\dots&\dots&\dots&\dots&\dots&\dots&\dots&\dots&\dots&\dots&\dots&\dots&\dots&\dots&\dots&\dots\\
			\hhline{~----------------}
			&$\textbf{\textit{d}}_{\substack{K+N-j+1}}$&$W_{1}$&\dots&$W_{1}$&\dots&$W_{1}$&$W_{1}$&\dots&$W_{1}$&\dots&$W_{1}$&$W_{1}$&\dots&\cellcolor{LimeGreen!50!white}$W_{N}$&\dots&$W_{1}$\\
			
			\hhline{~----------------}
			\ldelim\{{3}{10.5mm}[\parbox{3.5mm}{${\textbf{\textit{P}}_{j}}$}]&$\textbf{\textit{d}}_{\substack{K+N-j+2}}$&$W_{1}$&\dots&$W_{1}$&\dots&$W_{1}$&$W_{1}$&\dots&$W_{1}$&\dots&$W_{1}$&$W_{1}$&\dots&$W_{1}$&\dots&$W_{1}$\\
			\hhline{~----------------}
			&\dots&\dots&\dots&\dots&\dots&\dots&\dots&\dots&\dots&\dots&\dots&\dots&\dots&\dots&\dots&\dots\\
			\hhline{~----------------}
			&$\textbf{\textit{d}}_{\substack{K+2N-j}}$&$W_{1}$&\dots&$W_{1}$&\dots&$W_{1}$&$W_{1}$&\dots&$W_{1}$&\dots&$W_{1}$&$W_{1}$&\dots&$W_{N-1}$&\dots&$W_{1}$\\
			\hhline{~----------------}
			
			\ldelim\{{6}{10.5mm}[\parbox{3.5mm}{${\textbf{\textit{Q}}_{j}}$}]&$\textbf{\textit{d}}_{\substack{K+2N-j+1}}$&$W_{1}$&\dots&$W_{1}$&\dots&$W_{1}$&$W_{1}$&\dots&$W_{1}$&\dots&$W_{1}$&$W_{1}$&\dots&$W_{1}$&\dots&$W_{1}$\\
			\hhline{~----------------}
			&\dots&\dots&\dots&\dots&\dots&\dots&\dots&\dots&\dots&\dots&\dots&\dots&\dots&\dots&\dots&\dots\\
			\hhline{~----------------}
			&$\textbf{\textit{d}}_{\substack{K-i+1}}$&$W_{1}$&\dots&$W_{1}$&\dots&$W_{N-i}$&$W_{\substack{N-i+1}}$&\dots&\cellcolor{LimeGreen!50!white}$W_{N}$&\dots&$W_{N-i}$&$W_{\substack{N-i+1}}$&\dots&$W_{1}$&\dots&$W_{1}$\\
			\hhline{~----------------}
			\ldelim\{{3}{5mm}[\parbox{3.5mm}{${\textbf{\textit{B}}_{i}}$}]&$\textbf{\textit{d}}_{\substack{K-i+2}}$&$W_{1}$&\dots&$W_{1}$&\dots&$W_{\substack{N-i+1}}$&$W_{\substack{N-i+2}}$&\dots&$W_{1}$&\dots&$W_{\substack{N-i+1}}$&$W_{\substack{N-i+2}}$&\dots&$W_{1}$&\dots&$W_{1}$\\
			\hhline{~----------------}
			&\dots&\dots&\dots&\dots&\dots&\dots&\dots&\dots&\dots&\dots&\dots&\dots&\dots&\dots&\dots&\dots\\
			\hhline{~----------------}
			&$\textbf{\textit{d}}_{K}$&$W_{1}$&\dots&$W_{i-1}$&\dots&$W_{N-1}$&$W_{N}$&\dots&$W_{i-1}$&\dots&$W_{N-2}$&$W_{N-1}$&\dots&$W_{1}$&\dots&$W_{1}$\\
			\hhline{~----------------}
		\end{tabular}
		
		\caption{Demand set $\{\textbf{\textit{d}}_{l}:1\leq l\leq K\}$}
		\label{demands case 2}
	\end{table*}
	\noindent Consider the demand $\textbf{\textit{b}}_{l}$ defined as, 
	\begin{equation*}
	\textbf{\textit{b}}_{l}=\left\{\begin{aligned}
	\textbf{\textit{d}}_{N-l+1}, &\text{ for $1\leq l\leq N$}\\
	\textbf{\textit{d}}_{K+N-l+1}, &\text{ for $N+1\leq l\leq K$}
	\end{aligned}\right.
	\end{equation*}
	It can be noted that in demand $\textbf{\textit{b}}_{l}$, the user $U_l$ requests for the file $W_{N}$. The following lemma is easy to obtain:
	\begin{lemma}
		\label{case 2 reduction}
		Let $\textbf{S},\textbf{\textit{T}}\subset\{\textbf{\textit{d}}_{1}\dots,\textbf{\textit{d}}_{K}\}\setminus \{\textbf{\textit{b}}_{l}\}$ be such that for every demand in $\textbf{\textit{T}}$, user $U_l$ requests the file $W_{1}$. We have
		\begin{IEEEeqnarray*}{rl}
		H(W_{[N-1]},Z_{l},&X_{\textbf{S}})+\sum_{j\in \textbf{\textit{T}}}H(X_{j})+\frac{\mid\textbf{\textit{T}}\mid}{N}H(Z_{l})\geq H(W_{[N-1]},Z_{l},X_{\textbf{S}\cup \textbf{\textit{T}}})+\mid\textbf{\textit{T}}\mid
		\end{IEEEeqnarray*}
	\end{lemma}
	\begin{proof}
		We have,
		\begin{IEEEeqnarray*}{rl}
		H(W_{[N-1]}&,Z_{l},X_{\textbf{\textit{S}}})+\!\!\sum_{j\in \textbf{\textit{T}}}\!H(X_{j})+\frac{\mid\textbf{\textit{T}}\mid}{N}H(Z_{l})\\
		=&H(W_{[N-1]},Z_{l},X_{\textbf{\textit{S}}})+\sum_{j\in \textbf{\textit{T}}}H(X_{j})+\mid\textbf{\textit{T}}\mid H(Z_{l})-\frac{\mid\textbf{\textit{T}}\mid(N-1)}{N}H(Z_{l})\\
		\overset{(a)}{\geq}&H(W_{[N-1]},Z_{l},X_{\textbf{\textit{S}}})+\sum_{j\in \textbf{\textit{T}}}H(Z_{l},X_{j})-\frac{\mid\textbf{\textit{T}}\mid(N-1)}{N}H(Z_{l})\\
		\overset{(b)}{=}&H(W_{[N-1]},Z_{l},X_{\textbf{\textit{S}}})+\sum_{j\in \textbf{\textit{T}}}H(W_{1},Z_{l},X_{j})-\frac{\mid\textbf{\textit{T}}\mid(N-1)}{N}H(Z_{l})\\
		\overset{(a)}{\geq}&H(W_{[N-1]},Z_{l},X_{\textbf{\textit{S}}})+H(W_{1},Z_{l},X_{\textbf{\textit{T}}})+(\mid\textbf{\textit{T}}\mid-1)H(W_{1},Z_{l})-\frac{\mid\textbf{\textit{T}}\mid(N-1)}{N}H(Z_{l})\\
		\overset{(a)}{\geq}&H(W_{[N-1]},Z_{l},X_{\textbf{\textit{S}}\cup \textbf{\textit{T}}})+\mid\textbf{\textit{T}}\mid H(W_{1},Z_{l})-\frac{\mid\textbf{\textit{T}}\mid(N-1)}{N}H(Z_{l})\\
		\overset{(c)}{=}&H(W_{[N-1]},Z_{l},X_{\textbf{\textit{S}}\cup \textbf{\textit{T}}})+\frac{\mid\textbf{\textit{T}}\mid}{N}\Big(\sum_{i=1}^{N}H(W_{i},Z_{l})\Big)-\frac{\mid\textbf{\textit{T}}\mid(N-1)}{N}H(Z_{l})\\
		\overset{(a)}{\geq}&H(W_{[N-1]},Z_{l},X_{\textbf{\textit{S}}\cup \textbf{\textit{T}}})+\frac{\mid\textbf{\textit{T}}\mid}{N}\Big(H(W_{[N]},Z_{l})\Big)+ \Big(\frac{\mid\textbf{\textit{T}}\mid(N-1)}{N}H(Z_{l})\Big)-\frac{\mid\textbf{\textit{T}}\mid(N-1)}{N}H(Z_{l})\\
		\overset{(d)}{=}&H(W_{[N-1]},Z_{l},X_{\textbf{\textit{S}}\cup \textbf{\textit{T}}})+\frac{\mid\textbf{\textit{T}}\mid}{N}H(W_{[N]})=H(W_{[N-1]},Z_{l},X_{\textbf{\textit{S}}\cup \textbf{\textit{T}}})+\mid\textbf{\textit{T}}\mid
		\end{IEEEeqnarray*}
		\noindent where
		
		\hspace{-4mm}\begin{tabular}{cl}
			$(a)$& follows from submodularity property of entropy, \\$(b)$& follows from (\ref{I.1}), \\$(c)$& follows from (\ref{symmetry_file}), \\$(d)$& follows from (\ref{I.2}).
		\end{tabular}\\
	\end{proof}
	\noindent For $1\leq i \leq N$, let us consider the set of demands as shown below:
	\vspace{-5mm}
	\begin{table}[!ht]
		\centering
		\begin{tabular}{|c|c|c|c|}
			\hline
			Set&Demands&Number&Files Requested by  $U_i$\\
			\hline
			$\textbf{\textit{A}}_{i}$&$\textbf{\textit{d}}_{1}$, \dots, $\textbf{\textit{d}}_{N-i}$&$N-i$&$W_{i}$, \dots, $W_{N-1}$\\
			\hline
			$\textbf{\textit{B}}_{i}$&$\textbf{\textit{d}}_{K-i+2}$, \dots, $\textbf{\textit{d}}_{K}$&$i-1$&$W_{1}$, \dots, $W_{i-1}$\\
			\hline
			$\textbf{\textit{E}}_{i}$&$\textbf{\textit{d}}_{N+1}$, \dots, $\textbf{\textit{d}}_{2N-i}$&$N-i$&$W_{i}$, \dots, $W_{N-1}$\\
			\hline
			$\textbf{\textit{G}}_{i}$&$\textbf{\textit{d}}_{2N-i+1}$, \dots, $\textbf{\textit{d}}_{K-i+1}$&$K-2N+1$&$W_{1}$\\
			\hline
		\end{tabular}
	\end{table}

	\noindent These set are also indicated in TABLE \ref{demands case 2}. We also have a set of demands $\textbf{\textit{L}}_{i}$ defined as
	\begin{IEEEeqnarray}{rl}
	\label{case 2 ABCDE}
	{\textbf{\textit{L}}_{i}}&={\textbf{\textit{A}}_{i}}\cup {\textbf{\textit{B}}_{i}}\cup {\textbf{\textit{E}}_{i}}\cup{\textbf{\textit{G}}_{i}}.
	\end{IEEEeqnarray}
	Note that
	\begin{IEEEeqnarray}{rl}
	\label{case 2 null set}
	{\textbf{\textit{A}}_{N}}={\textbf{\textit{B}}_{1}}={\textbf{\textit{E}}_{N}}&=\phi\\
	\label{case 2 E}
	{\textbf{\textit{L}}_{i+1}}\cup \{\textbf{\textit{b}}_{{i+1}}\}&={\textbf{\textit{L}}_{i}}\\
	\label{case 2 B}
	\textbf{\textit{B}}_{i}\cup \{\textbf{\textit{b}}_{{N+i}}\}&=\textbf{\textit{B}}_{i+1}
	\end{IEEEeqnarray}
	\noindent It can be noted that in the demand set $\textbf{\textit{B}}_{i}$, both users $U_i$ and $U_{N+i}$ are requesting for the same set of files (for $1\leq i\leq N-1$). Note that $\mid \textbf{\textit{A}}_{i}\cup \textbf{\textit{B}}_{i}\mid=\mid \textbf{\textit{B}}_{i}\cup \textbf{\textit{E}}_{i} \mid=N-1$. Thus, we have
	\begin{align}
	\label{case2:A,B}
	M+(N-1)R\geq H(Z_{i})+H(X_{\textbf{\textit{A}}_{i}\cup \textbf{\textit{B}}_{i}})\geq H(Z_{i},X_{\textbf{\textit{A}}_{i}\cup \textbf{\textit{B}}_{i}})
	\end{align}
	Similarly,
	\begin{align}
	\label{case2:B,E}
	M+(N-1)R\geq H(Z_{i})+H(X_{\textbf{\textit{B}}_{i}\cup \textbf{\textit{E}}_{i}})\geq H(Z_{i},X_{\textbf{\textit{B}}_{i}\cup \textbf{\textit{E}}_{i}})
	\end{align}
	The following lemma is easy to obtain:
	\begin{lemma}
		\label{Lemma: case 2 ABEGL}
		The demand sets $\textbf{\textit{B}}_{i}$ and $\textbf{\textit{L}}_{i}$, defined as above, satisfy 
		\begin{IEEEeqnarray*}{rl}
		KM+(KN-2N+1)R\geq& H(W_{[N-1]},X_{\textbf{\textit{L}}_{N}})+\sum_{i=1}^{N-1}\bigg(H(W_{[N-1]},Z_{N+i},X_{\textbf{\textit{B}}_{i}})\bigg)+N(K-N)
		\end{IEEEeqnarray*}
	\end{lemma}
	\begin{proof}
		We have,
		\begin{IEEEeqnarray*}{rl}
			KM&+(KN-2N+1)R= N\left( M+(N-1)R +\frac{(K-2N+1)}{N}M+(K-2N+1)R\right)\\&+(N-1)(M+(N-1)R)\\
			\overset{(a)}{\geq}& \sum_{i=1}^{N}\bigg(H(Z_{i},X_{\textbf{\textit{A}}_{i}\cup\textbf{\textit{B}}_{i}})+\frac{(K-2N+1)}{N}H(Z_{i})+\sum_{l\in \textbf{\textit{G}}_{i}}H(X_{l})\bigg)+\sum_{i=1}^{N-1}\left(H(Z_{i},X_{\textbf{\textit{B}}_{i}\cup\textbf{\textit{E}}_{i}})\right)\\
			\overset{(b)}{=}& \sum_{i=1}^{N}\bigg(H(W_{[N-1]},Z_{i},X_{\textbf{\textit{A}}_{i}\cup\textbf{\textit{B}}_{i}})+\frac{(K-2N+1)}{N}H(Z_{i})+\sum_{l\in \textbf{\textit{G}}_{i}}H(X_{l})\bigg)\\&+\sum_{i=1}^{N-1}H(W_{[N-1]},Z_{i},X_{\textbf{\textit{B}}_{i}\cup\textbf{\textit{E}}_{i}})\\
			\overset{(c)}{\geq}& \sum_{i=1}^{N}\bigg(H(W_{[N-1]},Z_{i},X_{\textbf{\textit{A}}_{i}\cup\textbf{\textit{B}}_{i}\cup \textbf{\textit{G}}_{i}})+(K-2N+1)\bigg)+\sum_{i=1}^{N-1}H(W_{[N-1]},Z_{i},X_{\textbf{\textit{B}}_{i}\cup\textbf{\textit{E}}_{i}})\\
			\overset{}{=}&H(W_{[N-1]},Z_{N},X_{\textbf{\textit{A}}_{N}\cup\textbf{\textit{B}}_{N}\cup \textbf{\textit{G}}_{N}})+ \sum_{i=1}^{N-1}\bigg(H(W_{[N-1]},Z_{i},X_{\textbf{\textit{A}}_{i}\cup\textbf{\textit{B}}_{i}\cup \textbf{\textit{G}}_{i}})+H(W_{[N-1]},Z_{i},X_{\textbf{\textit{B}}_{i}\cup\textbf{\textit{E}}_{i}})\bigg)\\&+N(K-2N+1)\\
			\overset{(d)}{\geq}&H(W_{[N-1]},Z_{N},X_{\textbf{\textit{A}}_{N}\cup\textbf{\textit{B}}_{N}\cup \textbf{\textit{G}}_{N}\cup \textbf{\textit{E}}_{N}})+ \sum_{i=1}^{N-1}\bigg(H(W_{[N-1]},Z_{i},X_{\textbf{\textit{A}}_{i}\cup\textbf{\textit{B}}_{i}\cup\textbf{\textit{E}}_{i}\cup \textbf{\textit{G}}_{i}})+H(W_{[N-1]},Z_{i},X_{\textbf{\textit{B}}_{i}})\bigg)\\&+N(K-2N+1)\\
			\overset{(e)}{=}& \sum_{i=1}^{N}H(W_{[N-1]},Z_{i},X_{\textbf{\textit{L}}_{i}})+\sum_{i=1}^{N-1}H(W_{[N-1]},Z_{i},X_{\textbf{\textit{B}}_{i}})+N(K-2N+1)\\
			\overset{}{\geq}&\bigg(H(W_{[N-1]},X_{\textbf{\textit{L}}_{1}})+ \sum_{i=2}^{N}H(W_{[N-1]},Z_{i},X_{\textbf{\textit{L}}_{i}})\bigg)+\sum_{i=1}^{N-1}H(W_{[N-1]},Z_{i},X_{\textbf{\textit{B}}_{i}})+N(K-2N+1)\\
			\overset{(f)}{\geq}&(N-1)N+H(W_{[N-1]},X_{\textbf{\textit{L}}_{N}})+\sum_{i=1}^{N-1}H(W_{[N-1]},Z_{i},X_{\textbf{\textit{B}}_{i}})+N(K-2N+1)\\
			\overset{(g)}{=}& H(W_{[N-1]},X_{\textbf{\textit{L}}_{N}})+\sum_{i=1}^{N-1}H(W_{[N-1]},Z_{N+i},X_{\textbf{\textit{B}}_{i}})+N(K-N)
		\end{IEEEeqnarray*}
		
		\noindent where
		
		\hspace{-6mm}	\begin{tabular}{rl}
			$(a)$& follows from (\ref{case2:A,B}) and (\ref{case2:B,E}),\\
			$(b)$& follows from (\ref{I.1}) and definition of sets $\textbf{\textit{A}}_{i}$, $\textbf{\textit{B}}_{i}$ and $\textbf{\textit{E}}_{i}$,\\
			$(c)$& follows from Lemma \ref{case 2 reduction} with $\textbf{\textit{S}}=\textbf{\textit{A}}_{i}\cup \textbf{\textit{B}}_{i}$ and $\textbf{\textit{T}}=\textbf{\textit{G}}_{i}$,\\
			$(d)$& follows from the submodularity property of entropy and the fact that $\textbf{\textit{E}}_{N}=\phi$,\\ 
			$(e)$& follows from  (\ref{case 2 ABCDE}),\\ 
			$(f)$& follows from Lemma \ref{sum increment} with $\textbf{\textit{S}}_{i}=\textbf{\textit{L}}_{i}$, $l=1$, $j=N$ and (\ref{case 2 E}),\\
			$(g)$& follows from (\ref{symmetry}).
		\end{tabular}
		
	\end{proof}
	
	\noindent Now, for $2N\leq j \leq K$, consider another set of demands as shown below:
	\begin{table}[ht]
		\centering
		\begin{tabular}{|c|c|c|c|}
			\hline
			Set&Demands&Number&$\begin{array}{c}
			\text{Files Requested} \text{ by $U_{j}$}
			\end{array}$\\
			\hline
			$\textbf{\textit{P}}_{j}$&$\textbf{\textit{d}}_{K+N-i+2}$, \dots, $\textbf{\textit{d}}_{K+2N-j}$&$N-1$& $W_{1}$, \dots, $W_{N-1}$\\
			\hline
			$\textbf{\textit{Q}}_{j}$&$\textbf{\textit{d}}_{K+2N-j+1}$, \dots, $\textbf{\textit{d}}_{K}$&$j-2N$&$W_{1}$\\
			\hline
		\end{tabular}
	\end{table}

	\noindent These set are also indicated in TABLE \ref{demands case 2}. We also have a set of demands $\textbf{\textit{T}}_{j}$ defined as
	\begin{IEEEeqnarray}{rc}
	\label{case 2 PQT}
	{\textbf{\textit{T}}_{j}}=&{\textbf{\textit{P}}_{j}}\cup {\textbf{\textit{Q}}_{j}}
	\end{IEEEeqnarray}
	Note that
	\begin{IEEEeqnarray}{rl}
	\label{case 2 Q null}
	\textbf{\textit{Q}}_{2N}&=\phi\\
	\label{case 2 T}
	\textbf{\textit{T}}_{j+1}\cup \{{\textbf{\textit{b}}_{{j+1}}}\}&=\textbf{\textit{T}}_{j}\\
	\label{case 2 ET}
	{\textbf{\textit{T}}_{K}}\cup \{{\textbf{\textit{b}}_{{K}}}\}&={\textbf{\textit{L}}_{N}}\\
	\label{case 2 BT}
	\textbf{\textit{B}}_{N-1}\cup \{{\textbf{\textit{b}}_{{2N-1}}}\}&= \textbf{\textit{T}}_{2N}
	\end{IEEEeqnarray}
	Note that $\mid \textbf{\textit{P}}_{j}\mid=N-1$. Thus, we have
	\begin{align}
	\label{case2:P}
	M+(N-1)R\geq H(Z_{j})+H(X_{\textbf{\textit{P}}_{j}})\geq H(Z_{j},X_{\textbf{\textit{P}}_{j}})
	\end{align}
	The following lemma is easy to obtain:
	\begin{lemma}
		\label{Lemma: case 2 PQT}
		The demand set $\textbf{\textit{T}}_{j}$, defined as above, satisfy
		\begin{IEEEeqnarray*}{l}
			\frac{K(K-2N+1)}{2N}M+\frac{(K-2)(K-2N+1)}{2}R\geq \sum_{j=2N}^{K}\!H(W_{[N-1]},Z_{j},X_{\textbf{\textit{T}}_{j}})+\frac{(K-2N+1)(K-2N)}{2}
		\end{IEEEeqnarray*}
	\end{lemma}
	\begin{proof}
		We have,
		\begin{IEEEeqnarray*}{rl}
			\frac{K(K-2N+1)}{2N}&M+\frac{(K-2)(K-2N+1)}{2}R\\=&(K-2N+1)\left(M+(N-1)R\right)+\sum_{j=2N}^{K}\left(\frac{(j-2N)}{N}M+(j-2N)R\right)\\
			\overset{(a)}{\geq}& \sum_{j=2N}^{K}\bigg(H(Z_{j},X_{\textbf{\textit{P}}_{j}})+\frac{(j-2N)}{N}H(Z_{j})+\sum_{l\in \textbf{\textit{Q}}_{j}}H(X_{l})\bigg)\\
			\overset{(b)}{=}&\sum_{j=2N}^{K}\bigg(H(W_{[N-1]},Z_{j},X_{\textbf{\textit{P}}_{j}})+\frac{(j-2N)}{N}H(Z_{j})+\sum_{l\in \textbf{\textit{Q}}_{j}}H(X_{l})\bigg)\\
			\overset{(c)}{\geq}&\sum_{j=2N}^{K}\bigg(H(W_{[N-1]},Z_{j},X_{\textbf{\textit{P}}_{j}\cup\textbf{\textit{Q}}_{j}})+\frac{(j-2N)}{N}H(W_{[N]})\bigg)\\
			\overset{(d)}{=}&\sum_{j=2N}^{K}H(W_{[N-1]},Z_{j},X_{\textbf{\textit{T}}_{j}})+\frac{(K-2N+1)(K-2N)}{2}
		\end{IEEEeqnarray*}
		where
		
	\hspace{-6mm}	\begin{tabular}{rl}
			$(a)$& follows from (\ref{case2:P}) and the fact that $\textbf{\textit{Q}}_{2N}=\phi$,\\
			$(b)$& follows from (\ref{I.1}) and definition of set $\textbf{\textit{P}}_{j}$,\\
			$(c)$& follows from Lemma \ref{case 2 reduction} with $\textbf{\textit{S}}=\textbf{\textit{P}}_{j}$ and $\textbf{\textit{T}}=\textbf{\textit{Q}}_{j}$,\\
			$(d)$& follows from (\ref{case 2 PQT}).
		\end{tabular}\\
	\end{proof}

	\noindent Using the above lemma, we can obtain the following result:
	\begin{theorem}
		\label{Main_Result_2}
		For the $(N, K)$ cache network, when $1\leq N\leq\big\lceil\frac{K+1}{2}\big\rceil$, achievable memory rate pairs $(M,R)$ must satisfy the constraint
		\begin{equation*}
		\frac{K(K+1)}{2N}M+\frac{K(K-1)}{2}R\geq \frac{K^{2}+K-2}{2}
		\end{equation*}
	\end{theorem}
	\begin{proof}
		We have,
		\begin{IEEEeqnarray*}{rl}
			\frac{K(K+1)}{2N}&M+\frac{K(K-1)}{2}R=KM+(KN-2N+1)R+\frac{K(K-2N+1)}{2N}M+\frac{(K-2)(K-2N+1)}{2}R 
			\\
			\overset{(a)}{\geq}&H(W_{[N-1]},X_{\textbf{\textit{L}}_{N}})+\sum_{i=1}^{N-1}H(W_{[N-1]},Z_{N+i},X_{\textbf{\textit{B}}_{i}})+N(K-N)+\sum_{j=2N}^{K}H(W_{[N-1]},Z_{j},X_{\textbf{\textit{T}}_{j}})\\&+\frac{(K-2N+1)(K-2N)}{2}\\
			\overset{(b)}{=}& \left(H(W_{[N-1]},X_{\textbf{\textit{T}}_{K}},X_{\textbf{\textit{b}}_{K}})+\sum_{j=2N}^{K}\!H(W_{[N-1]},Z_{j},X_{\textbf{\textit{T}}_{j}})\right)+\sum_{i=1}^{N-1}\!H(W_{[N-1]},Z_{N+i},X_{\textbf{\textit{B}}_{i}})\\&+\frac{(K+1)(K-2N)+2N^{2}}{2}\\
			\overset{(c)}{\geq}& (K-2N+1)N+\left(H(W_{[N-1]},X_{\textbf{\textit{T}}_{2N}})+\sum_{i=1}^{N-1}H(W_{[N-1]},Z_{N+i},X_{\textbf{\textit{B}}_{i}})\right)\\&+\frac{(K+1)(K-2N)+2N^{2}}{2}\\
			\overset{(d)}{=}& \left(H(W_{[N-1]},X_{\textbf{\textit{B}}_{N-1}},X_{\textbf{\textit{b}}_{2N-1}})+\sum_{i=1}^{N-1}H(W_{[N-1]},Z_{N+i},X_{\textbf{\textit{B}}_{i}})\right)\\&+\frac{(K+1+2N)(K-2N)+2N(N+1)}{2}\\
			\overset{(e)}{\geq}& (N-1)N+H(W_{[N-1]},X_{\textbf{\textit{B}}_{1}})+\frac{(K+1+2N)(K-2N)+2N(N+1)}{2}\\
			\overset{(f)}{=}&H(W_{[N-1]})+\frac{K^{2}+K-2N}{2}
			\geq \frac{K^{2}+K-2}{2}
		\end{IEEEeqnarray*}
		where
		
	\hspace{-6mm}	\begin{tabular}{ll}
			$(a)$& follows from Lemma \ref{Lemma: case 2 ABEGL} and Lemma \ref{Lemma: case 2 PQT},\\  
			$(b)$& follows from (\ref{case 2 ET}), \\
			$(c)$& follows from (\ref{sum decrement}) with $\textbf{\textit{T}}_{i}=\textbf{\textit{T}}_{j}$, $l=2N$, $j=K$ and (\ref{case 2 T}),\\
			$(d)$& follows from (\ref{case 2 BT}), \\
			$(e)$& follows from (\ref{sum decrement}) with $\textbf{\textit{T}}_{i}=\textbf{\textit{B}}_{i}$, $l=1$, $j=N-1$ and (\ref{case 2 B}),\\
			$(f)$& follows from (\ref{case 2 null set}).
		\end{tabular}
	
	\end{proof}
The above observations improve upon the previous results from \cite{maddah2014fundamental,yu2016exact} as shown in TABLE \ref{Table:Case2}.
	\begin{table}[ht]
		\centering
		\begin{tabular}{|c|c|c|c|}
			\hline
			Memory &Rate \cite{maddah2014fundamental,yu2016exact}&Lower Bound\cite{maddah2014fundamental,yu2016exact}& New Lower Bound\\
			\hline&&&\\[-1.4em]
			$\frac{N(K-2)}{K}\leq M\leq \frac{N(K-1)}{K}$&$\frac{K^{2}+K-2}{K(K-1)}-\frac{(K+1)}{N(K-1)}M$&$R\geq 1-\frac{1}{N}M$&$R\geq\frac{K^{2}+K-2}{K(K-1)}-\frac{(K+1)}{N(K-1)}M$\\[0.3em]
			\hline
		\end{tabular}
	\caption{Rate memory tradeoff when $1\leq N\leq \big\lceil \frac{K+1}{2} \big\rceil$}
	\label{Table:Case2}
	\end{table}
	
	\noindent We summarise as:
	\begin{theorem}
		For the $(N,K)$ cache network, when $1\leq N\leq\big\lceil\frac{K+1}{2}\big\rceil$,  
		the exact rate memory tradeoff is given by
		\begin{equation}
		R^{*}(M)=\frac{K^{2}+K-2}{K(K-1)}-\frac{(K+1)}{N(K-1)}M
		\end{equation}
		where $M\geq \frac{N}{K}(K-2)$.
	\end{theorem}

	\section{Conclusions}
	In this paper, we considered the problem of characterizing the exact rate memory tradeoff for the canonical $(N, K)$ cache network, where we focused on the case of large caches. For $\big\lceil\frac{K+1}{2}\big\rceil\leq N\leq K$, in Section III, we proposed a new coded caching scheme and derived a matching lower bound leading to the characterization of the exact rate memory tradeoff when $M\geq\frac{N}{K}\left(K-2+\frac{(K-2+1/N)}{(K-1)}\right)$. For $1\leq N\leq \big\lceil\frac{K+1}{2}\big\rceil$, in Section IV, we derived a new lower bound matching the scheme proposed in \cite{yu2016exact}, thereby providing a characterization of the exact rate memory tradeoff when $M\geq \frac{N}{K}(K-2)$.

	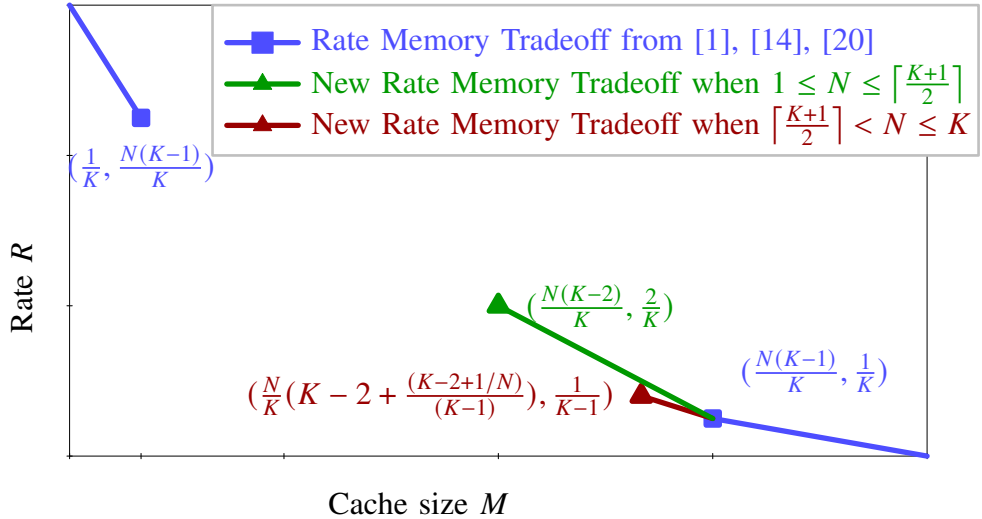
\begin{figure}[!ht]
		\centering
		\begin{tikzpicture}[line cap=round,line join=round,x=3.8cm,y=2cm]

\def\pic{\coordinate (O) at (0,0);
	\draw [ultra thin,step=3,black] (0,0) grid (3,3);
	\foreach \x in {0,1/4,3/4,6/4,9/4,3}
	\draw[shift={(\x,0)},color=black,thin] (0pt,1pt) -- (0pt,-1pt);
	\foreach \y in {0,1,2,3}
	\draw[shift={(0,\y)},color=black,thin] (1pt,0pt) -- (-1pt,0pt);
	\draw[color=black] (6cm,-18pt) node[left] { Cache size $M$};
	\draw[color=black] (-18pt,3.cm) node[left,rotate=90] { Rate $R$};
	
	\draw[smooth,blue!70,mark=otimes,samples=1000,domain=0.0:2.2,mark =$\otimes$,line width=2pt]
	{(0,3)--(1/4,9/4)  node[mark size=2.5pt,line width=2pt,label={below:$(\frac{1}{K},\frac{N(K-1)}{K})$}]{$\pgfuseplotmark{square*}$} (9/4,1/4)node[mark size=2.5pt,line width=2pt,label={above right:$(\frac{N(K-1)}{K},\frac{1}{K})$}]{$\pgfuseplotmark{square*}$}--(3,0)};
	\draw[smooth,red!60!black,samples=1000,domain=0.0:2.2,mark =$\otimes$,line width=2pt]
	{(2,0.4)node[mark size=4pt,line width=2pt,label={left:$(\frac{N}{K}(K-2+\frac{(K-2+1/N)}{(K-1)}),\frac{1}{K-1})$}]{\pgfuseplotmark{triangle*}}--(9/4,1/4)};
	
	\draw[smooth,green!60!black,samples=1000,domain=0.0:2.2,mark =$\otimes$,line width=2pt]
	{(3/2,1)node[mark size=4pt,line width=2pt,label={right:$(\frac{N(K-2)}{K},\frac{2}{K})$}]{$\pgfuseplotmark{triangle*}$}--(9/4,1/4)};

}

\pic

\draw [gray!50!white,line width=1pt,fill=white] (0.5,2)rectangle (3.18,3);
\begin{scope}[shift={(0.3,2.2)}] 
\draw [smooth,samples=1000,domain=0.0:2.2,red!60!black,mark=$\otimes$,line width=2pt] 
{(0.25,0) --node [mark size=4pt,line width=0.1pt]{$\pgfuseplotmark{triangle*}$} (0.5,0)}
node[right]{New Rate Memory Tradeoff when $\big\lceil\frac{K+1}{2}\big\rceil< N\leq K$};

\draw [yshift=0.75\baselineskip,smooth,samples=1000,domain=0.0:2.2,green!60!black,mark=$\otimes$,line width=2pt] 
{(0.25,0) --node [mark size=4pt,line width=0.1pt]{$\pgfuseplotmark{triangle*}$} (0.5,0)}
node[right]{New Rate Memory Tradeoff when $1\leq N\leq \big\lceil\frac{K+1}{2}\big\rceil$};

\draw [yshift=1.5\baselineskip,smooth,samples=1000,domain=0.0:2.2,blue!70,mark=$\otimes$,line width=2pt] 
{(0.25,0)--node [mark size=4pt,line width=0.1pt]{$\pgfuseplotmark{square*}$} (0.5,0)}
node[right]{Rate Memory Tradeoff from \cite{maddah2014fundamental,yu2016exact,chen2014fundamental}};

\end{scope}

\end{tikzpicture}
		\caption{Exact rate memory tradeoff for the $(N,K)$ cache network}
		\label{fig:(N,K)bound}
	\end{figure}

\end{document}